\documentclass[11pt]{article}
\usepackage{fullpage}
\usepackage{amssymb}
\usepackage{amsmath,amsfonts,amsthm,amssymb}
\usepackage{enumerate}
\usepackage{setspace}
\usepackage{fancyhdr}
\usepackage{lastpage}
\usepackage{extramarks}
\usepackage{algorithm}
\usepackage{comment}
\usepackage{algorithmic}
\usepackage{chngpage}
\usepackage{soul,color}
\usepackage{graphicx,float,wrapfig}
\usepackage[margin=1in]{geometry}
\usepackage[font=small]{caption}
\usepackage{subcaption}
\usepackage{epstopdf}

\usepackage[suppress]{color-edits}

\usepackage{tikz}
\addauthor{vs}{blue}
\addauthor{ct}{magenta}
\addauthor{ni}{red}
\addauthor{jm}{brown}
\addauthor{bl}{ForestGreen}

\newtheorem{theorem}{Theorem}[section]
\newtheorem{corollary}[theorem]{Corollary}

\newtheorem{definition}[theorem]{Definition}
\newtheorem{lemma}[theorem]{Lemma}
\newtheorem{remark}{Remark}[section]
\newtheorem{example}{Example}[section]

\newcommand{\Title}{
Combinatorial Assortment Optimization}

\title{\Title}

\author{
  Nicole Immorlica \\ Microsoft Research
  \and
  Brendan Lucier \\ Microsoft Research
  \and
  Jieming Mao \\ Princeton University
  \and
  Vasilis Syrgkanis \\ Microsoft Research
  \and
  Christos Tzamos \\ Microsoft Research
}

\begin{document}

  \begin{titlepage}
\maketitle \thispagestyle{empty}
\begin{abstract}
Assortment optimization refers to the problem of designing a slate of products to offer potential customers, such as stocking the shelves in a convenience store. The price of each product is fixed in advance, and a probabilistic choice function describes which product a customer will choose from any given subset.  We introduce the combinatorial assortment problem, where each customer may select a bundle of products.  We consider a model of consumer choice where the relative value of different bundles is described by a valuation function, while individual customers may differ in their absolute willingness to pay, and study the complexity of the resulting optimization problem.  We show that any sub-polynomial approximation to the problem requires exponentially many demand queries when the valuation function is XOS, and that no FPTAS exists even for succinctly-representable submodular valuations.  On the positive side, we show how to obtain constant approximations under a ``well-priced'' condition, where each product's price is sufficiently high.  We also provide an exact algorithm for $k$-additive valuations, and show how to extend our results to a learning setting where the seller must infer the customers' preferences from their purchasing behavior.  

 \end{abstract}
\end{titlepage}

\section{Introduction}

Imagine that you are an inventory manager, tasked with selecting which products to display on the shelves in a retail store.  These products are acquired from different producers, who control the suggested retail prices.  Your goal is to find a profitable assortment of items to offer, given a model of how customers choose which item(s) to ultimately purchase from the subset you display.  This \emph{assortment problem} captures a natural tradeoff.  If you offer only the most expensive items, then many customers might simply leave the store without purchasing anything.  On the other hand, a variety of inexpensive items might cannibalize sales from pricier goods and dilute the overall revenue.  Given a collection of possible items, and a model of customer preferences, which subset of items should you display to maximize revenue?

The assortment problem is of practical importance for brick and mortar stores, but is also relevant to online shopping platforms that must choose which products to display in response to a search query and whose price is exogenous, set by a third party.  Customers have limited patience and are more likely to select products from the first page of results, so the platform is incentivized to display a well-chosen slate of products.  Since an online platform may need to choose from a vast array of potential products, it is important to find computationally feasible solutions.

There is a growing literature on assortment in the field of revenue management, typically focusing on cases where each customer wants at most a single item.  In such unit-demand settings, the problem is captured by a \emph{choice function} that maps an assortment $S$ to a probability distribution describing which good in $S$ a customer will ultimately purchase.  Commonly-studied choice functions include multinomial logit functions~\cite{TalluriR2004}, exponential choice functions~\cite{AlptekinogluS2016}, and mixture models~\cite{BrontMDV2009}, among others.  
On the other hand, the computer science literature has mostly focused on combinatorial versions of revenue or welfare maximization when the designer controls the prices of items (see e.g. multi-dimensional revenue maximization \cite{Cai2012,Chawla2010,Haghpanah2015}) or the mode of interaction with the consumer (see e.g. combinatorial auctions \cite{Feige2009,Lehmann2001}). The important case of assortment optimization, where the platform designer is constrained to only design the set of available items, has been largely left untouched by the combinatorial optimization community. The goal of our work is to bridge this gap and explore the intersection of assortment and combinatorial optimization.

We introduce the \emph{combinatorial assortment problem}, where consumers may choose to purchase bundles of goods.  For example, a customer may want to buy a camera, possibly in combination with accessories, which may be either of the same brand as the camera or a cheaper off-brand variety.  These items may be complementary (a camera plus an accessory), or substitutes for each other (a brand-name accessory or a generic version of the same accessory).  We ask: \emph{given the relationship between the items for sale, and possibly a cardinality constraint on the number of items that can be shown, what is a revenue-maximizing selection to offer?}

We consider a model of consumer choice motivated by vertical customer differentiation.  In this model, the relationship between the items is fixed and common to all potential buyers, but customers vary in their willingness-to-pay.  Formally, the value that a buyer $i$ has for a certain bundle of goods $T$ is taken to be $w_i \cdot v(T)$, where $v$ is a valuation function common to all buyers and $w_i \geq 0$ is a buyer-specific multiplier that represent's the buyer's type.  This captures settings where the relative quality and relationship between the items is unambiguous, but customers vary in their ability to extract value from the items.  For example, if the items are cameras and accessories, a professional photographer might derive a value equal to $110\%$ of the reference value for any bundle, whereas an amateur might only derive $90\%$ of the reference value. Our market exhibits vertical differentiation in that all customers 
agree on the relative comparisons between bundles, so that if one bundle is more valuable and cheaper than another, everyone will buy the former.  In comparison, horizontally-differentiated choice functions like multinomial logit perturb the common component of the valuation by an additive constant; this causes customers to disagree on which bundles are more valuable, so that if one item is cheaper and has a higher common value than another, a positive fraction of customers would still prefer the latter.

\subsection{Our Results and Techniques}

\begin{figure}[t]
  \centering
  \begin{subfigure}[b]{0.47\textwidth}
  \centering
  	\begin{picture}(0,0)
\includegraphics[scale=0.6]{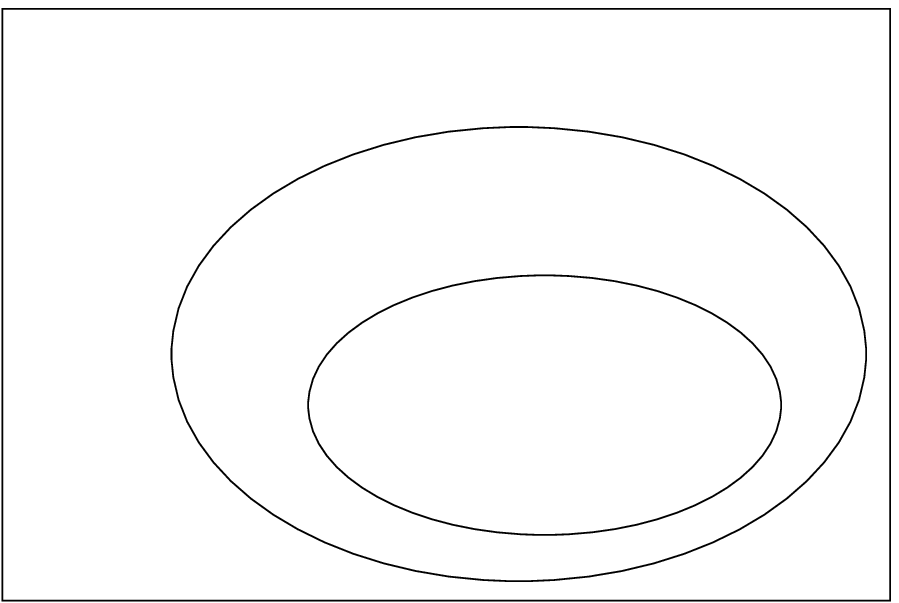}
\end{picture}
\setlength{\unitlength}{2367sp}
\begingroup\makeatletter\ifx\SetFigFont\undefined
\gdef\SetFigFont#1#2#3#4#5{
  \reset@font\fontsize{#1}{#2pt}
  \fontfamily{#3}\fontseries{#4}\fontshape{#5}
  \selectfont}
\fi\endgroup
\begin{picture}(4285,2864)(817,-2483)
\put(2750,-1281){\makebox(0,0)[lb]{\smash{{\SetFigFont{8}{14.4}{\rmdefault}{\mddefault}{\updefault}{\color[rgb]{0,0,0}$\mathrm{K-Add}$}
}}}}
\put(3558,-714){\makebox(0,0)[lb]{\smash{{\SetFigFont{8}{14.4}{\rmdefault}{\mddefault}{\updefault}{\color[rgb]{0,0,0}$\mathrm{SM}$}
}}}}
\put(3990, 91){\makebox(0,0)[lb]{\smash{{\SetFigFont{8}{14.4}{\rmdefault}{\mddefault}{\updefault}{\color[rgb]{0,0,0}$\mathrm{XOS}$}
}}}}
\put(350,101){\makebox(0,0)[lb]{\smash{{\SetFigFont{6}{14.4}{\rmdefault}{\mddefault}{\updefault}{\color[rgb]{0,0,0}$\Omega(n^{0.5-\epsilon})$ communication lower bound}
}}}}
\put(350,-204){\makebox(0,0)[lb]{\smash{{\SetFigFont{6}{14.4}{\rmdefault}{\mddefault}{\updefault}{\color[rgb]{0,0,0}$\mathrm{NP}$-hard (succinct)}
}}}}
\put(1393,-812){\makebox(0,0)[lb]{\smash{{\SetFigFont{6}{14.4}{\rmdefault}{\mddefault}{\updefault}{\color[rgb]{0,0,0}$\mathrm{NP}$-hard (2-demand)}
}}}}
\put(1953,-1652){\makebox(0,0)[lb]{\smash{{\SetFigFont{6}{14.4}{\rmdefault}{\mddefault}{\updefault}{\color[rgb]{0,0,0}Exact algorithm}
}}}}
\end{picture}
   	\caption{Arbitrary prices.}
  \end{subfigure}
  ~
  \begin{subfigure}[b]{0.47\textwidth}
  \centering
  	\begin{picture}(0,0)
\includegraphics[scale=0.6]{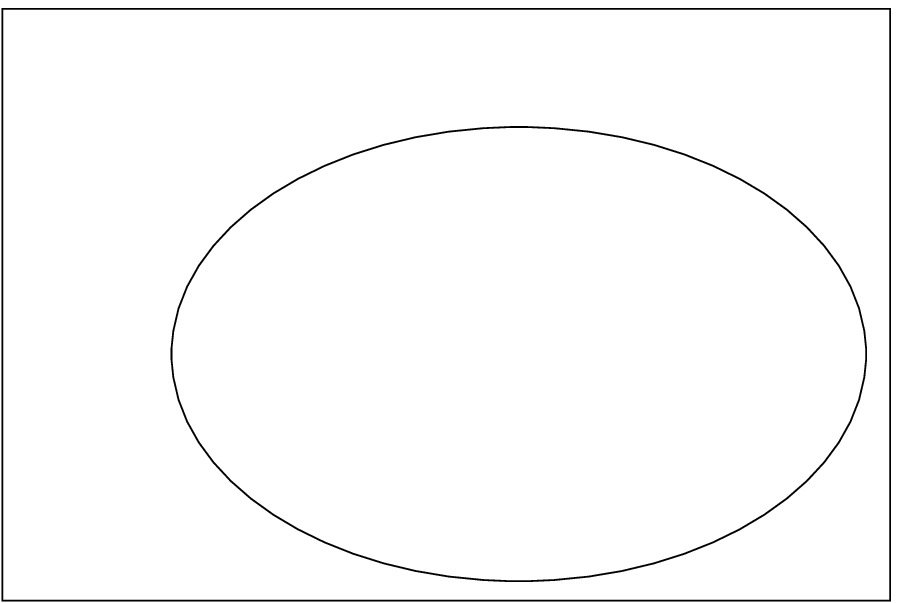}
\end{picture}
\setlength{\unitlength}{2367sp}
\begingroup\makeatletter\ifx\SetFigFont\undefined
\gdef\SetFigFont#1#2#3#4#5{
  \reset@font\fontsize{#1}{#2pt}
  \fontfamily{#3}\fontseries{#4}\fontshape{#5}
  \selectfont}
\fi\endgroup
\begin{picture}(4285,2864)(717,-2483)
\put(3918, 98){\makebox(0,0)[lb]{\smash{{\SetFigFont{8}{14.4}{\rmdefault}{\mddefault}{\updefault}{\color[rgb]{0,0,0}$\mathrm{ANY}$}
}}}}
\put(3436,-714){\makebox(0,0)[lb]{\smash{{\SetFigFont{8}{14.4}{\rmdefault}{\mddefault}{\updefault}{\color[rgb]{0,0,0}$\mathrm{GS}$}
}}}}
\put(368,-39){\makebox(0,0)[lb]{\smash{{\SetFigFont{6}{14.4}{\rmdefault}{\mddefault}{\updefault}{\color[rgb]{0,0,0}$O(1)$-$\mathrm{APX}$ (w/o constraints)}
}}}}
\put(1455,-1356){\makebox(0,0)[lb]{\smash{{\SetFigFont{6}{14.4}{\rmdefault}{\mddefault}{\updefault}{\color[rgb]{0,0,0}$O(1)$-$\mathrm{APX}$ (w/ constraints)}
}}}}
\end{picture}
   	\caption{Well-priced items.}
  \end{subfigure}
  \caption{\small Computational landscape of Combinatorial Assortment Optimization.  For arbitrary prices, the negative results for XOS and SM (submodular) valuations hold even without cardinality constraints, and the exact algorithm for K-ADD ($k$-additive) valuations applies even with cardinality constraints.  For the case of well-priced items (see Section~\ref{sec:structural}), we give a constant-approximate algorithm for general valuations without cardinality constraints, or for GS (gross substitutes) valuations with cardinality constraints.}
\label{fig:landscape}
\end{figure}

Our goal is to explore the computational complexity of combinatorial assortment. We will characterize the limits of polynomial time computation or approximability and provide conditions under which simple heuristics such as a greedy algorithm, or exhaustively searching over small assortments, are optimal or approximately so. Interestingly, we will see a stark difference in the computational landscape, depending on how well the items are priced (with respect to the distribution over consumer types). 
It turns out that assuming the item prices are not too low can make otherwise computationally hard assortment problems easy to solve or approximate (see Figure~\ref{fig:landscape}).

The bulk of our results apply in the case where the valuation function $v$ is known to the assortment planner, and the types $w_i$ are unknown but drawn from a known prior distribution $\mathcal{F}$.  We then investigate the difficulty of the assortment problem as a function of the structural assumptions imposed on the valuation $v$. At the end of the paper we extend many of our algorithmic results to a setting where the planner must learn these parameters from samples.

\paragraph{Negative Results} Our main results are summarized in Figure~\ref{fig:landscape}.  We begin by showing that, in general, the combinatorial assortment problem is inherently difficult.  Even in the 
deterministic case, where all buyers have the exact same preferences and these are known to the optimizer (i.e., the type distribution $\mathcal{F}$ is a point mass at $1$),
it is hard to approximate the revenue of the optimal assortment to a factor of $o(n^{1/2 - \varepsilon})$ for any constant $\varepsilon > 0$, where $n$ is the number of items to choose from.  This is true even if there is no constraint on the number of items to be shown, and even if the valuation function is an XOS function, a subclass of subadditive functions.\footnote{A valuation is subadditive if, for any sets of items $S$ and $T$, $v(S \cup T) \leq v(S) + v(T)$.  A valuation is XOS if it is the maximum of a collection of additive functions.}  Notably, this is a class of valuations where the welfare maximization problem can be well-approximated~\cite{Lehmann2001,Feige2009}.

This hardness result takes the form of a communication complexity bound, independent of any computational hardness assumptions.  We show that an approximation algorithm requires an exponential amount of communication with an oracle that can answer demand queries about the valuation function $v$.  Note that it is too much to hope for a lower bound in a fully general model of communication with a valuation oracle, since in particular the oracle could simply communicate the optimal assortment, which can be described in polynomially many bits.  Instead, our proof considers a communication model in which information about the valuation $v$ is split between two oracles, and show that exponential communication between the oracles is necessary to obtain any reasonable approximation.  We then show how the pair of oracles can simulate a demand query oracle.  One implication of this result is that any assortment algorithm with a sub-polynomial approximation factor requires exponentially many demand queries about the valuation function $v$. 

We next show that even for valuation functions that can be described succinctly,\footnote{Formally: an XOS valuation that is the maximum of only $2$ additive functions.} it is still NP-hard to compute the optimal assortment.  Like the communication complexity result, this holds even if all buyers have type $1$.  If we move beyond this deterministic case and allow buyer types to be drawn from an arbitrary distribution, then we show that there is no FPTAS for the combinatorial assortment problem with XOS (or even submodular) valuations even if each customer wants at most two items.\footnote{When customers demand at most 2 items, the XOS condition is equivalent to submodularity.  A valuation is submodular if, for any sets of items $S$ and $T$, $v(S \cup T) + v(S \cap T) \leq v(S) + v(T)$.  This is equivalent to each item having diminishing marginal value, and is more restrictive than subadditivity.}  Furthermore, the natural greedy heuristics that adds items to the assortment one by one, maximizing the marginal revenue increase on each step, fails to obtain a constant approximation for submodular valuations, even in the deterministic case where $\mathcal{F}$ is a point mass.

\paragraph{Algorithmic Results} Motivated by these lower bounds, we characterize settings in which natural methods achieve good approximations, and where exact solutions can be computed in polynomial time.
We first characterize settings where displaying all items is a good approximation to the optimal revenue.  As mentioned earlier, offering all items might be highly suboptimal in the presence of ``cheap'' items that might cannabilize sales from more profitable items.  We show that such an issue is inherently due to items being sold at too low a price.  We say that the goods are ``well-priced'' if, roughly speaking, the price of each bundle is at least its optimal (i.e., Myerson) reserve price, in a world where only that bundle is for sale.  When goods are substitutes, this is equivalent to each individual item's price being at least its Myerson price.  This may be the case if the individual product retailers are behaving like monopolists and not responding to the assortment planner, such as when the platform is driving only a small portion of the producer's overall revenue.  We show that if the goods are well-priced, and the type distribution satisfies the standard regularity property, then offering all items is a $4$-approximation to the optimal revenue.

\begin{theorem}
For combinatorial assortment with well-priced items and regular type distribution, the assortment that selects all items is a $4$-approximation to the optimal expected revenue.
\end{theorem}

We also show that if there is a cardinality constraint on the number of items that can be shown, then greedily accepting items to maximize marginal revenue also yields a constant approximation when the valuations satisfy the \emph{gross substitutes} condition, which is a stronger notion of substitutability than submodularity.  

\begin{theorem}
For cardinality-constrained combinatorial assortment with well-priced items, a gross substitutes valuation, and regular type distribution, the assortment that selects items greedily by revenue is a $\tfrac{4e}{e-1}$-approximation to the optimal expected revenue.
\end{theorem}

In addition to these approximation results, we present an exact algorithm for combinatorial assortment when the valuation function is $k$-demand additive.  That is, when each buyer desires at most $k$ items, and the value for such a bundle is the sum of the individual item values.  This class extends unit-demand valuations to bundles of more than a single item.  For this setting, we describe a dynamic programming solution that runs in time $O(n^{2k})$.  Our solution builds an optimal assortment by first optimizing for high-type buyers and incrementally modifying the assortment to cater to lower types.  This algorithm does not require any assumptions about items being well-priced, and applies whether or not there are cardinality constraints on the assortment.

Finally, for $k$-demand valuations that may not be additive, we show that under a certain revenue-concavity assumption on the type distribution, the optimal assortment will have size at most $k$.  

\paragraph{Extension: Welfare Maximization} We conclude by considering two extensions.  First, we note that most of our positive results apply also to the goal of maximizing welfare, rather than maximizing revenue.  The welfare maximization problem is still non-trivial, since the presence of cheap goods can result in lower-valued items being purchased.  However, we show that if items are well-priced then offering all items is, in fact, the welfare-optimal assortment.  Note that this is a stronger result than for revenue-maximization, where we established a $4$-approximation.  Under a cardinality constraint, the greedy algorithm for assortment yields a $\frac e {e-1}$ approximation to the optimal welfare for well-priced items and gross substitutes valuations.  Finally, our dynamic program for additive $k$-demand valuations applies just as well to the welfare objective, and can be used to compute a welfare-optimal assortment.  Kleinberg et al.~\cite{KleinbergMU2017} study the learnability of a class of comparison-based choice functions.

\paragraph{Extension: Learning} The second extension concerns a setting where $v$ and $\mathcal{F}$ are not known to the seller.  Rather, the seller must learn these through demand queries: repeatedly choosing a slate of items and observing a buyer's choice.  We show that the dynamic programming solution for $k$-demand additive valuations can be implemented in this learning setting, with the loss of an $O(k\epsilon)$ additive error factor, using $\Theta(n^{k+1} \log(n) / \epsilon^2 )$ queries.

\subsection{Related Work}

There is a growing literature on (unit-demand) assortment optimization in the management science literature.  Talluri and van Ryzin~\cite{TalluriR2004} provide a closed-form solution when buyer choices follow the multinomial logit model.  Rusmevichientog et al.~\cite{RSS2010} extend this solution to the case of cardinality-restricted assortment, and Davis et al.~\cite{DavisGT2014} show how to solve for the optimal assortment under more general nested logit models.  When the choice function is described by a mixture of multinomial logit models, the assortment problem is NP-hard but various integer programming methods and approximation algorithms are known~\cite{BrontMDV2009,DesirG2015,MendezDiazBVZ2014}.

There has also been work studying learning in assortment, where the product slate can be adjusted to learn customer preferences.  Caro and Gallein~\cite{CaroG2007} consider learning in a model of assortment without substitution effects, where the demand for each product is unaffected by the other products in the assortment.  Ulu et al.~\cite{UluHA2012} study the dynamic learning problem when products exhibit purely horizontally differentiation, as modeled by location on a line segment.  Agrawal et al.~\cite{AgrawalAGZ2016} consider a multi-armed bandit model of dynamic assortment, and show how to achieve near-optimal regret for multinomial logit choice models.  Kleinberg et al.~\cite{KleinbergMU2017} consider a general class of comparison-based choice models, and study the complexity of learning their model from samples.

The combinatorial assortment problem can be viewed as a restricted form of mechanism design, where the design space consists only of choosing which subset of items to display.  This is more restrictive than sequential posted pricing, where the designer can also choose the price at which each item can be sold (e.g., \cite{Chawla2010}).

\section{The Combinatorial Assortment Optimization Problem}

There is a set $N$ of $n$ items. Each item $i$ has a fixed price $p_i \geq 0$. We assume items are indexed so that $p_1 \leq p_2 \leq \cdots \leq p_n$.  There is an unbounded supply (i.e., number of copies) of each item.

There is a collection of buyers, each of whom wish to purchase a subset of the items.  Each buyer $j$ has a value $u_j(S) = w_j \cdot v(S)$ for each subset $S \in [n]$ of goods.  Here $v(S)$ is a common valuation that determines the relationship between the goods, for all buyers, and $w_j$ is a buyer-specific scaling factor.  We refer to $w_j$ as the \emph{type} of buyer $j$.  We assume that each $w_j$ is sampled independently from a distribution $\mathcal{F}$, which we refer to as the type distribution.  We sometimes also call $w_j$ the multiplicative noise of buyer $j$.
When $\mathcal{F}$ is a point mass on 1 (i.e., $u_j = v$ for each buyer $j$), we call the problem \emph{noiseless}.  We call the general problem \emph{noisy}.

Given a subset of items $T$ displayed to a buyer $j$, the buyer will pick $S \subseteq T$ maximizing $u_j(S) - \sum_{i \in S} p_i$ and pay $\sum_{i \in S} p_i$. Our goal as a seller is to pick an \emph{optimal assortment}, which is a subset $T$ of at most $\ell$ items that maximizes the expected revenue. 
Here $\ell$ is a parameter of the problem.  
We will focus first on the \emph{unconstrained} case of $\ell = n$, then consider general $\ell$ in Section~\ref{sec:constrained}.
For most of the paper we will assume that $v$ and $\mathcal{F}$ are known to the seller and given as inputs to the optimization problem.  In Section~\ref{sec:learn} we relax this assumption and suppose $v$ and $\mathcal{F}$ are fixed but unknown to the seller, who must learn about them by interacting with buyers.

\paragraph{Valuation classes.}
We focus on variants of the combinatorial assortment problem where the valuation function $v$ lies in a given class.  We assume that valuations are monotone non-decreasing and normalized so that $v(\emptyset) = 0$.  In this paper we will focus on the following valuation classes, which encode forms of substitutability between items.

\begin{itemize}
\item \textbf{additive:} there exist $v_1, \dotsc, v_n \geq 0$ such that $v(S)=\sum_{i \in S} v_i$.
\item \textbf{XOS:} there exist additive valuations (i.e., clauses) $v_1,...,v_m$ such that $v(S)=\max_{i \in [m]} v_i(S)$.  
\item \textbf{submodular:} for all $S, T \subseteq [n]$, $v(S \cup T) + v(S \cap T) \leq v(S) + v(T)$. 
\item \textbf{gross substitutes:} 
for all $S,T \subseteq [n]$ and $x \in S$, one of the following is true:\footnote{We use the M\#-exchange characterization of gross substitutes, since it will be convenient for our proofs~\cite{Leme}.}
\begin{enumerate}
\item $v(S) + v(T) \leq v(S \backslash \{ x\}) + v(T \cup \{x \})$. 
\item There exists $y \in T$, $v(S)  + v(T) \leq v(S \backslash \{x \} \cup \{y\}) + v(T \backslash \{y\} \cup \{x\})$. 
\end{enumerate}
\end{itemize}

We will also be interested in valuations that encode a constraint that a buyer does not derive benefit from receiving more than a certain number of items.

\begin{definition}
Valuation $v$ is $k$-demand if, for all $S \subseteq N$, $v(S)= \max_{T \subseteq S, |T|\leq k} v(T)$.  That is, the buyer derives no benefit from receiving more than $k$ items. 
We say that valuation $v$ is additive (resp.\ XOS, submodular) $k$-demand if there is an additive (resp.\ XOS, submodular) valuation $v'$ such that, for all $S \subseteq N$, $v(S) = \max_{T \subseteq S, |T| \leq k} v'(T)$.
\end{definition}

We note that these valuation classes can be ordered from most to least restrictive, as follows:
\emph{Additive $k$-demand $\subseteq$ gross substitutes $\subseteq$ submodular $\subseteq$ XOS.  }

\section{Hardness of Combinatorial Assortment}
\label{sec:hardness}

In this section we explore the hardness of the Combinatorial Assortment problem.  We give a general hardness of approximation result for XOS valuations, even in the noiseless setting.  We then show that even when valuations can be succinctly represented, the problem remains NP-hard.  We also demonstrate that even when valuations are submodular, the natural greedy heuristic fails to obtain a good approximation.  All missing proofs can be found in Appendix \ref{app:hardness}.

\paragraph{Hardness of approximation, even without noise.}

We begin by considering the noiseless setting, where $\mathcal{F}$ is a point mass at $1$ and hence the valuation of the buyer is known exactly.  
Our first result shows that for XOS valuations, the combinatorial nature of the problem leads to strong hardness of approximation.  Indeed, it may take exponential many demand queries to achieve better than an $O(\sqrt{n})$-approximation to the combinatorial assortment problem.

\begin{theorem}
\label{thm:XOShard}
For XOS valuations, any $o(n^{1/2 - \varepsilon})$-approximate algorithm for the combinatorial assortment problem 
requires $\Omega(\exp ({n^{2\varepsilon}/24})/n)$ demand queries.
\end{theorem}

Note that Theorem~\ref{thm:XOShard} is a query complexity bound, and puts no limitations on the algorithm's running time.
Theorem \ref{thm:XOShard} can be extended to a more general statement about communication complexity under a certain query model. See Remark \ref{rm:cc} for details. The general result will suggest that combinatorial assortment problem is hard to approximate with a sub-exponential number of a certain class of queries. 
Note that we cannot hope for Theorem~\ref{thm:XOShard} to extend to a fully general communication complexity bound with an arbitrary query model: if arbitrary queries are allowed, one could directly ask for the optimal assortment, which can be succinctly described.

The proof of Theorem \ref{thm:XOShard} follows by reducing from the communication complexity of the equality function to the combinatorial assortment problem.  Two players, Alice and Bob, play a communication game where they each hold an (exponentially-long) input string and want to determine if they hold the same string.  They each use their input strings to construct XOS function clauses, and the input to the combinatorial assortment problem will be the XOS valuation function containing both Alice and Bob's clauses.
Each of Alice's clauses corresponds to a large set of items, and assigns small values; each of Bob's corresponds to a small set, and assigns large values.  The buyer will only ever buy a set of items corresponding to one of these clauses.  The optimizer would prefer that the buyer chooses one of Alice's large sets.  However, the clauses are constructed so that if Alice and Bob's inputs are equal, then each of Alice's clauses is ``dominated'' by one of Bob's clauses, so there is no assortment where the buyer purchases many items.  However, if the inputs are unequal, then at least one of Alice's clauses is ``uncovered,'' and the corresponding items would be purchased if they were the only items available.
By carefully designing the XOS clauses in this way, we can show that approximation of the combinatorial assortment problem will also solve the equality problem.  

\paragraph{Hardness for succinct valuations.}

Theorem~\ref{thm:XOShard}'s hardness is a communication bound, and relies on the fact that an XOS function may require exponentially many bits to fully describe.  As we now show, the combinatorial assortment problem remains hard even for XOS valuations with succint descriptions.  In particular, the problem is NP-hard, again in the noiseless setting, even if we restrict to valuations with only two clauses (i.e., the maximum of two additive functions).  

\begin{theorem}
\label{thm:XOS2hard}
For any XOS valuations with only 2 clauses, finding the optimal revenue is NP-hard in the noiseless case and the offline setting. 
\end{theorem}

The idea of the proof is to relate the optimal revenue of the combinatorial assortment problem to the solution to a knapsack problem, implementing the knapsack constraints by comparing values between the two clauses in the combinatorial assortment problem. 

\paragraph{Hardness for $2$-demand valuations in the noisy setting.}
\label{sec:2demandhard}

One might also wonder if the hardness results above are driven by the large sets of goods desired by the buyers.  What if we restrict attention to $k$-demand buyers, where $k$ is a small constant?  

One observation is that in the noiseless setting, the optimal assortment for a $k$-demand valuation will contain at most $k$ items, so the problem can be solved in time $n^{O(k)}$ by evaluating the revenue for all subsets of size $k$.  So this question is interesting only in the more general noisy setting.  

Theorem \ref{thm:2demandhard} shows that even for submodular 2-demand valuations, there can be no FPTAS for combinatorial assortment.  Therefore, we can only hope to get an efficient algorithm for $k$-demand valuations if we add add more restrictions, for example, to require the valuations to also be additive.

\begin{theorem}
\label{thm:2demandhard}
For submodular $2$-demand valuations, it is NP-hard to approximate the optimal revenue within approximation factor $1+1/n^c$, for some large enough constant $c$.  In particular, there is no FPTAS in this setting unless $P = NP$. 
\end{theorem}

The proof is a reduction from the $k$-clique problem. Given a graph, we construct a 2-demand valuation and a distribution $\mathcal{F}$ over the types. We embed the edge information into the prices of pairs of items. $\mathcal{F}$ is carefully chosen such that an assortment has large revenue if and only if it corresponds to a set of vertices which form a $k$-clique in the original graph. 

\paragraph{Greedy assortment fails for submodular valuations.}

We've shown that there is no FPTAS for submodular valuations in the general noisy setting.  One might wonder if it's possible to obtain a constant approximation, however, by using a simple heuristic.  One natural idea for submodular valuations is to use a greedy approach: repeatedly add the revenue-maximizing item to the assortment, until either no item remains or until adding any one item causes revenue to decrease.  The following example shows that this heuristic can lead to approximation $\Omega(n)$, even without noise.

\begin{example}
There are $n = m+1$ items, which we'll label $\{0, 1, \dotsc, m\}$.  The valuation $v$ is: 
\[ v(S) = \begin{cases} 
m|S| & \text{ if $0 \not\in S$ }\\
m+(m-1)|S| & \text{ otherwise }
\end{cases}
\]
One can verify that this valuation is indeed submodular.  Suppose $p_0 = m$ and $p_j = m-2$ for all $j > 0$.  The greedy algorithm selects item $0$ first, as it generates revenue $m$ which is larger than $m-2$, the revenue from any other single item.  However, having selected item $0$, the greedy algorithm would not add more items, since if the assortment is $\{0,i\}$ for any $i > 0$, the buyer would choose to buy only item $i$ leading to a loss of revenue.  So  greedy obtains revenue $m$.  The optimal assortment takes all items other than $0$, for a revenue of $(m-2)m$.
\end{example}
 
\section{Structural and Algorithmic Results}

\paragraph{Approximate assortment for well-priced items.}\label{sec:structural}

As mentioned in the introduction, it can be highly suboptimal to select all items in the combinatorial assortment problem, since the presence of a cheap but valuable item might cannibalize revenue from more expensive items.
One might wonder, then, if such a situation can be made less severe if the items are all priced ``reasonably.''  For example, suppose that each individual item is assigned the price that would maximize revenue when that item is sold by itself.  Indeed, we would argue that such prices are very reasonable if the items are typically sold separately, and it is precisely the assortment platform that presents these items in combination with each other.  We will show that under such an assumption, plus a regularity assumption on the type distribution, it is approximately revenue-optimal to show all of the items.
Let us first define formally the assumptions needed for our result.
\begin{definition}[Regularity]
We say that type distribution $F$ is \emph{regular} if the virtual value function $\phi(w) = w - \frac{1-F(w)}{f(w)}$ is non-decreasing, where $f$ denotes the density function of distribution $F$.
\end{definition}

Regularity is a common assumption in the revenue maximization literature.  Many natural distributions are regular, including uniform, gaussian, and exponential distributions.  

\begin{definition}[Revenue Curve]
The \emph{revenue curve} of a type distribution $F$ is $R(p) = p(1 - F(p))$.
\end{definition}

We can think of $R(p)$ as describing the revenue obtained if we were to offer a single item with value $1$ and price $p$ to a buyer whose type is drawn from distribution $F$. As we show the total revenue of an assortment can be expressed as a function of $R$. 

The \emph{optimal reserve} (or \emph{Myerson reserve}) for $F$ is the value $r$ that maximizes $R(r)$ (or the supremum over such values $r$, if the maximum is not unique).

\begin{definition}[Well-priced]
Suppose type distribution $F$ is regular with Myerson reserve $r$ and non-increasing density after $r$.\footnote{In fact, our results for well-priced combinatorial assortment hold for distributions that satisfy a weaker condition than regularity.  It is enough for the well-pricedness condition to hold for some value $r$ (not necessarily the Myerson reserve) such that the density function $f$ is non-increasing after $r$, and the revenue curve $R$ is non-increasing after $r$.}
Then the combinatorial assortment problem with type distribution $F$ is \emph{well-priced} if, for each subset of items $S \subseteq N$,
$$\sum_{i \in S} p_i \geq r \cdot v(S).$$
\end{definition}

We think of $r$ as a desired threshold on the type of buyers who purchase items.  For example, if we focus on a single item $i$, then reserve $r$ corresponds to a price of $r \cdot v_i$, as this is the price at which a buyer with type $w > r$ would choose to purchase.
The well-priced condition requires that the price assigned to any set of items is at least the reserve $r$, scaled appropriately to the value of the set.  Note that if $v$ is subadditive, it is enough for each individual item to be well-priced as this implies the condition for any larger set of items as well.

We show that for well-priced instance of the combinatorial assortment problem, selecting all items yields a $4$-approximation to the optimal revenue. The proof can be found at Appendix \ref{sec:convex-proof}.

\begin{theorem}\label{thm:convex}
Choosing $S = N$ is a $4$-approximation to the optimal revenue for well-priced combinatorial assortment.
\end{theorem}

The idea behind Theorem~\ref{thm:convex} is to show that the revenue curve $R$ can be well-approximated by a modified revenue curve $\hat{R}$ that is convex on the range $[r,\infty)$.  We show that for convex curves, maximizing revenue reduces to the problem of maximizing utility, and hence the (modified) revenue is maximized by the assortment that maximizes buyer utility, which is to display all items.

\paragraph{Exact assortment when revenue is concave.}\label{sec:concave}

This approximation result used intuition that when revenue curves are convex, it is preferable to show as many items as possible.  As it turns out, the reverse intuition holds as well: if the revenue curve is concave, then it is preferable to show fewer items.  In particular, if buyers are $k$-demand, then the optimal assortment will consist of at most $k$ items. The proof of Theorem \ref{thm:concave} can be found in Appendix \ref{app:concave}.

\begin{theorem}\label{thm:concave}
Suppose that buyers are $k$-demand, 
and the revenue function 
$R$ is concave over the support of the type distribution $F$.
Then there exists an optimal assortment $S$ with $|S| \leq k$.
\end{theorem}

Recall that in Section~\ref{sec:2demandhard} we noted that, in general, the optimal assortment for $k$-demand buyers may contain far more than $k$ items.  In particular, a heuristic that simply enumerates all assortments of size at most $k$ will not find an optimal solution in general.  Theorem~\ref{thm:concave} shows that such a heuristic does find an optimal solution in cases where the revenue curve is concave.  

\begin{example}\label{ex:uniform}
Suppose that buyers are $k$-demand with uniform type distribution over $[a,b]$. The revenue curve $R(w) = w\cdot\min\left\{ 1, \frac {b-w} {b-a}\right\}$ is concave for all $w \in [0,b]$ and thus by Theorem~\ref{thm:concave} the optimal assortment consists of at most $k$ items. 
\end{example}
\begin{remark}
If in Example~\ref{ex:uniform} items are well-priced, Theorem~\ref{thm:convex} implies that even though the optimal assortment is small, showing all items yields a $4$-approximation to the optimal revenue. 
\end{remark}

\paragraph{Exact combinatorial assortment for $k$-demand buyers.}\label{sec:alg}
We showed in Section~\ref{sec:2demandhard} that the combinatorial assortment problem is hard even for submodular $2$-demand buyers.  We instead turn to additive $k$-demand valuations, which include unit-demand valuations as a special case.  We show that for constant $k$, there is a polynomial-time algorithm that solves the combinatorial assignment problem. The proof of Theorem \ref{thm:dpmain} can be found in Appendix \ref{sec:dp}. Importantly, this result applies even in the general noisy cases where buyer values are not fully known in advance.  Like the submodular $2$-demand case considered in Section~\ref{sec:2demandhard}, the optimal assortment for additive $k$-demand valuations may include many more than $k$ items.

\begin{theorem}
\label{thm:dpmain}
For additive $k$-demand valuations, there exists an algorithm that finds the revenue-optimal assortment in time\footnote{Our algorithms depend on the type distribution, which may be continuous.  The runtime bound assumes that the CDF of this distribution can be queried in $O(1)$ time.  See Appendix \ref{app:computation} for a detailed discussion.} $O(n^{2k} + n^2\log(n))$.
\end{theorem}

The algorithm we propose is a dynamic program, which incrementally builds an optimal assortment by considering how the purchasing behavior of a buyer changes with $w$.  
To build intuition for our dynamic program, consider first the unit-demand case of $k=1$.  In this case, each buyer chooses at most a single item to purchase.  The utility derived by purchasing item $i$ is $w v_i - p_i$, which we can plot as a line mapping $w$ to utility.  A choice of assortment $S$ then corresponds to a subset of $n$ possible lines; and for any given value of $w$, the item with the highest utility would be chosen.  We can think of this as tracing the maximum over this set of lines; see Figure~\ref{fig:upper-envelope}. Given this pictorial representation our DP algorithm computes the optimal revenue ``from left to right'' adding lines/items to the assortment but only keeping track of the last line that was added. 

In the case $k>1$, we are interested in the revenue obtained by tracing the top $k$ lines given an assortment. Computing the optimal revenue is inherently more difficult in this case, as it doesn't suffice to store only the top $k$ lines at any given point $w$ (see Figure \ref{fig:k=2} in Appendix \ref{sec:dp}). However, we are able to extend our DP by showing that any lines that were among the top $k$ earlier but are not in the top $2k-1$ at the current point $w$ won't be among the top $k$ lines for any $\hat w > w$. This allows us to only keep track of the top $2k-1$ lines, resulting in $\tilde{O}(n^{2k})$ runtime.

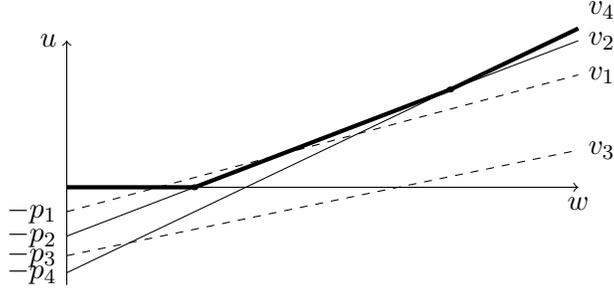
\begin{figure}\begin{center}
\begin{tikzpicture}[xscale=1.7,yscale=1.3]

\draw [->] (0,0) -- (4,0)  node [below] {$w$};
\draw [->] (0,-1) -- (0,1.5) node [left] {$u$};

\draw [fill] (3,1) circle [radius=0.025];
\draw [fill] (1,0) circle [radius=0.025];

\draw [dashed, domain=0:4] plot (\x, 1.4*\x/4-.2*5/4) node [right] {$v_1$};
\draw [domain=0:4] plot (\x, 2.0*\x/4-.4*5/4) node [right] {$v_2$};
\draw [dashed, domain=0:4] plot (\x, .27*\x-.7) node [right] {$v_3$};
\draw [domain=0:4] plot (\x, 2.5*\x/4-.7*5/4) node [above right] {$v_4$};

\draw (0,-.2*5/4) node [left] {$-p_1$};
\draw (0,-.4*5/4) node [left] {$-p_2$};
\draw (0,-.7*5/4) node [left] {$-p_4$};
\draw (0,-.7) node [left] {$-p_3$};

\draw [ultra thick, domain=0:1] plot (\x, 0);
\draw [ultra thick, domain=1:3] plot (\x, 2.0*\x/4-.4*5/4);
\draw [ultra thick, domain=3:4] plot (\x, 2.5*\x/4-.7*5/4);

\end{tikzpicture}
\end{center}
\vspace{-10pt}
\caption{The dark solid line represents a unit-demand buyer's utility for the assortment $S=\{2,4\}$. It is the upper envelope of the lines corresponding to items 2 and 4.}
\label{fig:upper-envelope}
\end{figure}

\section{Extensions}\label{sec:extensions}

\paragraph{Constrained assortment.}\label{sec:constrained}

To this point we focused exclusively on the case of unconstrained assortment, where $\ell = n$.  For general $\ell$, the lower bounds from Section~\ref{sec:hardness} still apply.  Also,
the dynamic program for exact revenue-optimal assortment for additive $k$-demand valuations 
solves the constrained case;
one need only track the remaining budget for additional items as part of the program.  
See Corollary \ref{cor:dp} for a detailed proof.

\begin{theorem}
\label{thm:dpmainconstrained}
For additive $k$-demand valuations and any cardinality constraint $\ell$, there exists an algorithm that finds the revenue-optimal assortment of at most $\ell$ items in time $O(n^{2k} \ell)$.
\end{theorem}

Theorem~\ref{thm:convex} specified conditions under which it is approximately optimal to select all items.  Under a cardinality constraint, this solution may not be feasible.  However, if the buyer valuations are gross substitutes, a greedy assortment algorithm is approximately optimal.  
The idea is to reduce from revenue maximization to utility maximization, as in Theorem~\ref{thm:convex}, then note that the total utility derived from the buyers is a submodular function.  See Appendix~\ref{subsec:convex.constrained.gross} for details.

\begin{theorem}\label{thm:convex.constrained.gross}
For gross substitutes valuations and well-priced items, a $(6.33)$-approximation to the revenue-optimal assortment of size at most $\ell$ can be computed in time $\tilde O(\ell^3 n)$.
\end{theorem}

\paragraph{Welfare maximizing assortment.}\label{sec:welfare}

We have focused on revenue-maximization, but 
assortment optimization for welfare maximization is also non-trivial. 
The presence of cheap items in the assortment can reduce the total welfare and should be excluded.
We note that the algorithm we developed for revenue-maximization under additive $k$-demand valuations can be easily adjusted for welfare maximization (see Remark \ref{rm:dp} for details).  Also, if items are well-priced, our results for revenue maximization apply to welfare maximization with even better constants.  In particular, for unconstrained assortment, selecting the slate of all items is welfare-optimal if items are well-priced.  See Appendix~\ref{app:welfare}.

\paragraph{Learning assortments from demand samples.}\label{sec:learn}
Suppose $v$ and $\mathcal{F}$ are not known to the seller.  Instead, the algorithm 
can learn about $v$ and $\mathcal{F}$
via samples, taken by choosing a slate of items to sell and observing a buyer's choice.   Details appear in Appendix~\ref{app:learn}.

We show how to implement our dynamic program for $k$-additive valuations in this learning setting, by characterizing the algorithm's robustness to noise.  We show that if the algorithm can make $\Theta(n^{k+1}\log(n)/\varepsilon^2)$ queries, then our dynamic programming solution will be within an $O(\varepsilon \cdot \max_{|S|=k} \sum_{i \in S} p_i)$ additive factor to the optimal revenue.

We also show that a variant of Theorem \ref{thm:concave} applies to the learning setting.  This requires choosing the best of a polynomial number of assortments.  Since the highest revenue is bounded, standard concentration arguments imply that we can evaluate the revenue of any given assortment to within a small additive error by making polynomially many queries.
 
\bibliographystyle{plain}
\bibliography{bib}

\begin{thebibliography}{10}

\bibitem{AgrawalAGZ2016}
Shipra Agrawal, Vashist Avadhanula, Vineet Goyal, and Assaf Zeevi.
\newblock A near-optimal exploration-exploitation approach for assortment
  selection.
\newblock In {\em Proceedings of the 2016 ACM Conference on Economics and
  Computation}, EC '16, pages 599--600, New York, NY, USA, 2016. ACM.

\bibitem{AlptekinogluS2016}
Aydın Alptekinoğlu and John~H. Semple.
\newblock The exponomial choice model: A new alternative for assortment and
  price optimization.
\newblock {\em Operations Research}, 64(1):79--93, 2016.

\bibitem{BrontMDV2009}
Juan Jos{\'e}~Miranda Bront, Isabel M{\'e}ndez-D\'{\i}az, and Gustavo Vulcano.
\newblock A column generation algorithm for choice-based network revenue
  management.
\newblock {\em Oper. Res.}, 57(3):769--784, May 2009.

\bibitem{Cai2012}
Yang Cai, Constantinos Daskalakis, and S.~Matthew Weinberg.
\newblock Optimal multi-dimensional mechanism design: Reducing revenue to
  welfare maximization.
\newblock In {\em Proceedings of the 2012 IEEE 53rd Annual Symposium on
  Foundations of Computer Science}, FOCS '12, pages 130--139, Washington, DC,
  USA, 2012. IEEE Computer Society.

\bibitem{CaroG2007}
Felipe Caro and J{\'e}r{\'e}mie Gallien.
\newblock Dynamic assortment with demand learning for seasonal consumer goods.
\newblock {\em Management Science}, 53:276--292, 2007.

\bibitem{Chawla2010}
Shuchi Chawla, Jason~D. Hartline, David~L. Malec, and Balasubramanian Sivan.
\newblock Multi-parameter mechanism design and sequential posted pricing.
\newblock In {\em Proceedings of the Forty-second ACM Symposium on Theory of
  Computing}, STOC '10, pages 311--320, New York, NY, USA, 2010. ACM.

\bibitem{DavisGT2014}
James~M. Davis, Guillermo Gallego, and Huseyin Topaloglu.
\newblock Assortment optimization under variants of the nested logit model.
\newblock {\em Operations Research}, 62(2):250--273, April 2014.

\bibitem{DesirG2015}
Antoine Desir and Vineet Goyal.
\newblock Near-optimal algorithms for capacity constrained assortment
  optimization.
\newblock {\em Technical Report, Department of Industrial Engineering and
  Operations Research, Columbia University}, 2015.

\bibitem{Feige2009}
Uriel Feige.
\newblock On maximizing welfare when utility functions are subadditive.
\newblock {\em SIAM Journal on Computing}, 39(1):122--142, 2009.

\bibitem{Haghpanah2015}
Nima Haghpanah and Jason Hartline.
\newblock Reverse mechanism design.
\newblock In {\em Proceedings of the Sixteenth ACM Conference on Economics and
  Computation}, pages 757--758. ACM, 2015.

\bibitem{KleinbergMU2017}
Jon Kleinberg, Sendhil Mullainathan, and Johan Ugander.
\newblock Comparison-based choices.
\newblock In {\em Proceedings of the 2017 ACM Conference on Economics and
  Computation}, EC '17, pages 127--144, New York, NY, USA, 2017. ACM.

\bibitem{Kushilevitz:1996:CC:264772}
Eyal Kushilevitz and Noam Nisan.
\newblock {\em Communication Complexity}.
\newblock Cambridge University Press, New York, NY, USA, 1997.

\bibitem{Lehmann2001}
Benny Lehmann, Daniel Lehmann, and Noam Nisan.
\newblock Combinatorial auctions with decreasing marginal utilities.
\newblock In {\em Proceedings of the 3rd ACM Conference on Electronic
  Commerce}, EC '01, pages 18--28, New York, NY, USA, 2001. ACM.

\bibitem{Leme}
Renato~Paes Leme.
\newblock Gross substitutability: An algorithmic survey.

\bibitem{MendezDiazBVZ2014}
Isabel M{\'e}ndez-D\'{\i}az, Juan~Jos{\'e} Miranda-Bront, Gustavo Vulcano, and
  Paula Zabala.
\newblock A branch-and-cut algorithm for the latent-class logit assortment
  problem.
\newblock {\em Discrete Appl. Math.}, 164:246--263, February 2014.

\bibitem{RSS2010}
Paat Rusmevichientong, Zuo-Jun~Max Shen, and David~B. Shmoys.
\newblock Dynamic assortment optimization with a multinomial logit choice model
  and capacity constraint.
\newblock {\em Operations Research}, 58(6):1666--1680, 2010.

\bibitem{TalluriR2004}
Kalyan Talluri and Garrett van Ryzin.
\newblock Revenue management under a general discrete choice model of consumer
  behavior.
\newblock {\em Management Science}, 50(1):15--33, 2004.

\bibitem{UluHA2012}
Canan Ulu, Dorothée Honhon, and Aydın Alptekinoğlu.
\newblock Learning consumer tastes through dynamic assortments.
\newblock {\em Operations Research}, 60(4):833--849, 2012.

\end{thebibliography}

\appendix

\section{A Note on Computation}\label{app:computation}

The algorithm of Theorem~\ref{thm:dpmain}, as well as all the algorithms presented in this work, depend on the distribution of types which may be continuous. The only assumption required for their runtime is that the CDF of this distribution can be queried in $O(1)$ time. As the CDF may be a real number, our algorithms require a real RAM model where basic calculus can be performed as a single operation. This assumption can be easily dropped if the CDF oracle returns only the CDF within accuracy of $\frac \varepsilon {2 n} 2^{-B}$, where $B$ is the number of bits required to represent valuations and prices. As the revenue of an assortment $S$ equals $\sum_{S' \subseteq S} \text{Prob[$S'$ is bought]} \cdot \sum_{i \in S'} p_i$ and all probabilities are accurate within $\frac \varepsilon {n} 2^{-B}$, this results in an additive error in revenue calculation of at most $\frac \varepsilon {n} 2^{-B} \sum_{i=1}^n p_i  \le \varepsilon$. 

Additionally, revenue calculations depend only on the CDF at points where consumers are indifferent about the set of items to purchase. These are points of the form $\frac {\sum_{i \in S\setminus S'} p_i} {v(S) - v(S')}$ for $S' \subset S$. Since we assume both valuations and prices can be represented using $B$-bit numbers, the algorithms require that the CDF is only queried in rational numbers where both numerator and denominator are $B$-bit integers.

The algorithm of Theorem~\ref{thm:convex.constrained.gross} requires computing a modified revenue curve which is convex so that the resulting optimization problem is submodular. The optimal such curve corresponds to the lower convex envelope of the actual revenue curve $R(w) = w(1-F(w))$. While this curve might be easily computable in closed form in some cases, it cannot be computed efficiently using only query access to the CDF. To overcome this issue, we first preprocess the revenue curve by rounding it into powers of $(1+\varepsilon)$. This results in a different revenue maximization problem whose solution is close to the original one. Moreover, this revenue curve consists of very few pieces and can be listed explicitly and thus computing the convex envelope and all revenue calculations can be done efficiently. See details in Appendix~\ref{subsec:convex.constrained.gross}. 
\section{Missing Proofs of Section \ref{sec:hardness}}
\label{app:hardness}

\subsection{Hardness results in the noiseless setting}
\begin{lemma}
\label{lem:set}
For any constant $\varepsilon >0$, there exist $M= \exp ({n^{2\varepsilon}/24})$ sets $X_i$'s which have sizes $n^{1/2}$ and are subsets of $[n]$ such that 
\[
\forall 1 \leq i < j \leq M, |X_i \cap X_j| \leq n^{\varepsilon}.
\]
\end{lemma}

\begin{proof}
We are just going to pick $M$ random subset $X_i$ of size $n^{1/2}$. Then we will show the probability that $\forall 1 \leq i < j \leq , |X_i \cap X_j| \leq n^{\varepsilon}$ is positive. 

Fix some pair $(i,j)$. Define random variable $Z_p$ to be one if $p \in X_i$ and $p \in X_j$. Otherwise $Z_p$ will be 0. We have
\[
\Pr [Z_p = 1] = \frac{1}{n^{1/2}} \cdot \frac{1}{n^{1/2}} = \frac{1}{n}.
\]
Although $Z_p$'s are not independent, they are negative correlated. We can apply the multiplicative Chernoff bound:
\[
\Pr[\sum_{p=1}^n Z_p > n^{\varepsilon}] \leq \exp(-(n^{\varepsilon}-1)^2/3) < \exp(-n^{2\varepsilon}/12).
\]

Then by a Union bound over all pairs $(i,j)$, we have
\[
\Pr[\forall 1 \leq i < j \leq M, |X_i \cap X_j| \leq n^{\varepsilon}] > 1- \binom{M}{2}\cdot \exp(-n^{2\varepsilon}/12) > 0. 
\]
\end{proof}

\begin{theorem}[Restatement of Theorem \ref{thm:XOShard}]
For XOS valuations, any algorithm (no restrictions on the running time) which approximates the optimal revenue within factor smaller than $n^{1/2-\varepsilon}/2$ needs $\Omega(\exp ({n^{2\varepsilon}/24})/n)$ demand queries in the noiseless case. 
\end{theorem}

\begin{proof}
By Lemma \ref{lem:set}, we can find $M= \exp ({n^{2\varepsilon}/24})$ sets $X_i$'s which have sizes $n^{1/2}$ and are subsets of $[n]$ such that 
\[
\forall 1 \leq i < j \leq M, |X_i \cap X_j| \leq n^{\varepsilon}.
\]

Let $b =n^{1/2-\varepsilon}/2$, For each $X_i$, define $Y_{i,1}, ....,Y_{i,b}$ to be an arbitrary partition of $X_i$, i.e., 
\begin{enumerate}
\item $Y_{i,1} \cup \cdots \cup Y_{i,b} = X_i$
\item $Y_{i,j} \cap Y_{i,j'} = \emptyset$ for all $1 \leq j < j'  \leq b$. 
\item $|Y_{i,j}| = |X_i|/b = 2n^{\varepsilon} $ for all $j \in [b]$.  
\end{enumerate}

Let $W,W'$ be a subset of $[M]$. Define the XOS valuation $v_{W,W'}$ as the following by specifying its clauses:
\begin{enumerate}
\item For each $i$ in $W$, $v_W$ has clause $c_i$ such that $c_i(\{j\}) = 2$ for all $j \in X_i$ and $c_i(\{j\}) = 0$ for all $j \not\in X_i$. These clauses are called c-clauses.
\item For each $i$ in $W'$ and each $j \in [b]$, $v_W$ has clause $d_{i,j}$ such that $d_{i,j}(\{j'\}) = b + 2$ for $j' \in Y_{i,j}$ and $d_{i,j}(\{j'\}) = 0$ for $j' \not\in Y_{i,j}$. These clauses are called d-clauses. 
\end{enumerate}
Suppose there's an algorithm $A$(can be randomized) which guarantees better than $n^{1/2-\varepsilon}/2$-approximation on the optimal revenue for XOS valuations. Let's assume $A$ uses $Q$ queries in expectation. 

Now consider the following communication problem: 
\begin{enumerate}
\item Alice gets input $W_X$ which is a subset of $[M]$ and $|W_X| = M/2$. And Bob gets input $W_Y$ which is a subset of $[M]$ and $|W_Y| = M/2$.
\item The goal is just to decide whether $W_X = W_Y$ by communication between Alice and Bob.
\end{enumerate}
This problem is just equality problem, which has zero-error randomized communication complexity $\Omega(\log \binom{M}{M/2}) = \Omega(M)$ (see Example 3.9 of \cite{Kushilevitz:1996:CC:264772}). 

Now consider the following protocol $\pi$ for the communication problem based on algorithm $A$:\
\begin{enumerate} 
\item Alice and Bob run algorithm $A$ locally on valuation $v_{W_X,W_Y}$ and item price 1 (i.e. $p_i = 1, \forall i\in [n]$). They use public randomness if $A$ needs randomness. Notice that without any information of $v_{W_X,W_Y}$ they can run most part of $A$ except demand queries. They are going to simulate the value queries by the following communication procedure. 
\item Whenever $A$ is making a demand query, Alice sends buyer's favorite subset over all the clauses she knows (i.e. $c_i$'s for all $i \in W_X$). Alice also sends the utility of that subset. Bob does the similar thing. Then they can figure out the result of the demand query by picking the subset with better utility. 
\item After running $A$, if the output of $A$ is larger than $2n^{\varepsilon}$, output ``not equal''. Otherwise output ``equal''.
\end{enumerate}

Since Alice and Bob use $O(n)$ bits of communication for each demand query, the expected communication complexity of $\pi$ is $O(Q \cdot n)$. 

Now let's prove $\pi$ correctly solves equality. 
\begin{enumerate}
\item When $W_X = W_Y$, we are going to show that the optimal revenue is at most $2n^{\varepsilon}$. Let's assume the optimal revenue is achieved by the seller showing $T$ and the buyer buying $S \subseteq T$. 
\begin{enumerate}
\item If $v_{W_X,W_Y}(S)$ is evaluated on some c-clause, let's assume the clause is $C_i$. Then we know $S \subseteq X_i$. Since $W_X =W_Y$, we know $i \in W_Y$. Since $Y_{i,1},...,Y_{i,b}$ is a partition of $X_i$, we know there exists $j$ such that $|Y_{i,j} \cap S| \geq |S|/b$. The utility of buying $Y_{i,j} \cap S$ is at least $(|S|/b) \cdot (b+2-1) > |S|$. On the other hand, the utility of buying $S$ is $|S|$. We get a contradiction now and therefore it's never the case that $v_{W_X,W_Y}(S)$ is evaluated on some c-clause.
\item If $v_{W_X,W_Y}(S)$ is evaluated on some d-clause, then we know $|S|$ will be smaller than $2n^{\varepsilon}$ as any d-clause are non-zero on $2n^{\varepsilon}$ items. Therefore in this case the revenue is at most $2n^{\varepsilon}$. 
\end{enumerate}
\item When $W_X \neq W_Y$, as $|W_X| = |W_Y|$, there exists $i$ such that $i \in W_X$ and $i \not\in W_Y$. We are going to show that if the seller show subset $X_i$, the buyer will buy the entire set. And therefore the optimal revenue is $n^{1/2}$. Formally, we will show for any subset $S \subseteq X_i$, the utility of the buying $S$ is at most the utility of buying $X_i$ for the buyer. It is clear that the utility of buying $X_i$ is $C_i(X_i) - |X_i| = |X_i| = n^{1/2}$. For any $S \subseteq X_i$,
\begin{enumerate}
\item If $v_{W_X,W_Y}(S)$ is evaluated on some c-clause, then the utility of buying $S$ is at most $|S| \leq n^{1/2}$. The equality is only achieved when $S = X_i$. 
\item If $v_{W_X,W_Y}(S)$ is evaluated on some d-clause, let that d-clause be $d_{i',j}$. Since $i \in W_Y$, we know $i' \neq i$, therefore $|Y_{i',j}\cap X_i| \leq |X_{i'} \cap X_i| \leq n^{\varepsilon}$. Therefore the utility of buying $S$ is at most $n^{\varepsilon} \cdot ( n^{1/2-\varepsilon}/2 + 2 -1) < n^{1/2}$. 
\end{enumerate}
\end{enumerate}
Since $A$ guarantees better than $(n^{1/2-\varepsilon}/2)$-approximation, when $W_X \neq W_Y$, it would output something larger than $2n^{\varepsilon}$. And when $W_X = W_Y$, it will output something at most $2n^{\varepsilon}$. Therefore protocol $\pi$ correctly solves equality. Then by the communication lower bound, we know that $Q\cdot n = \Omega(M)$. Therefore $Q = \Omega(\exp ({n^{2\varepsilon}/24})/n)$.
\end{proof}

\begin{remark}
\label{rm:cc}
It's easy to check that the proof of Theorem \ref{thm:XOShard} still works if we switch demand queries with other queries which can be computed by Alice and Bob using polynomial many bits of communication. For example, a value query can be simulated by $O(\log(n))$ bits as Alice and Bob can just report the value of the set in their own parts and then take the maximum.
\end{remark}
\begin{lemma}
\label{lem:opt}
Among all the subsets that achieve the optimal revenue for the seller, there exist a subset $T$ such that if the seller shows $T$, the buyer would pick all the items in $T$. 
\end{lemma}

\begin{proof}
Let $T'$ be an arbitrary set that achieves the optimal revenue for the seller. Assume in this case, buyer chooses $S \subseteq T'$. Then if we just show the buyer $S$, the buyer will still pick $S$. 
\end{proof}

\begin{theorem}[Restatement of Theorem \ref{thm:XOS2hard}]
For any XOS valuations with only 2 clauses, finding the optimal revenue is NP-hard in the noiseless case. 
\end{theorem}

\begin{proof}
We will reduce from the knapsack problem which is NP-hard, i.e.:
 \begin{alignat*}{2}
    \textbf{maximize }   & \sum_{i=1}^m v_i x_i  \\
    \textbf{subject to ~~~~} & \sum_{i=1}^m w_ix_i \leq W\\
                       & x_i \in \{0,1\}\\
  \end{alignat*}
Suppose we have an algorithm $A$ to find the optimal revenue for XOS valuations with 2 clauses. We are going to solve the knapsack problem in the following way:

We set $n = m+1$. We then construct XOS valuation with 2 clauses $c_1$ and $c_2$. We set $c_1$, $c_2$ and price $p_i$'s as:
\begin{enumerate}
\item For $1 \leq i \leq m = n-1$, $c_1(\{i\}) = v_i+1$, $c_2(\{i\}) = w_i + v_i+1$ and $p_i = v_i$.
\item For $i = n$, $c_1(\{i\}) = v_1 + \cdots + v_m +1 + W$, $c_2(\{i\}) = 0$ and $p_i = v_1 + \cdots +v_m+1$. 
\end{enumerate}

Notice that the optimal revenue is at least $p_n$ because the seller can always show only item $n$. And as $p_n > p_1 + \cdots + p_{n-1}$, to achieve the optimal revenue, the seller needs to make sure the buyer purchases item $n$, which means the value needs to be evaluated on $c_1$. 

Now we are going to show the optimal revenue is equal to the maximum objective of the knapsack problem plus $p_n$. 

\begin{enumerate}
\item By Lemma \ref{lem:opt}, let $T$ be the set that when the seller shows $T$, the buyer would buy $T$ and the seller gets optimal revenue. As discussed before $n \in T$. Then we know that $c_1(T) -\sum_{i\in T} p_i \geq c_2(T \backslash n)-\sum_{i\in T \backslash n} p_i $. This implies $\sum_{i \in T \backslash n} w_i \leq W$. Therefore, if we pick $x_i = 1$ iff $i \in T$, we have $ \sum_{i=1}^m w_ix_i \leq W$ and $\sum_{i=1}^m v_i x_i$ at least the optimal revenue minus $p_n$. Therefore the optimal revenue is at most the maximum objective of the knapsack problem plus $p_n$.
\item Now consider the optimal solution for the knapsack problem $x_i$'s. Define $T = \{i| x_i = 1\} \cup \{n\}$. It's easy to check that the buyer would pick $T$ if the seller shows $T$ because $\sum_{i=1}^m w_ix_i \leq W$. Therefore the optimal revenue is at least the maximum objective of the knapsack problem plus $p_n$.
\end{enumerate}
We finish the proof by noticing the fact that the decision version of the knapsack problem is NP-hard.
\end{proof}
 
\subsection{Hardness for $2$-demand valuations in the noisy setting}
\begin{theorem}[Restatement of Theorem \ref{thm:2demandhard}]
For submodular $2$-demand valuations, it's NP-hard to approximate the optimal revenue within approximation factor $1+1/n^c$ in the noisy case for some large enough constant $c$. Therefore there's no FPTAS to compute the optimal revenue in this setting unless $P = NP$. 
\end{theorem}

\begin{proof}
The reduction is from $k$-clique. Let the $k$-clique instance be graph $G = (V,E)$. Let $|V| = n_G$. Now consider the following mathematical program:
 \begin{alignat*}{2}
    \textbf{maximize }   & \sum_{(i,j)\in E} x_ix_j -n_G^2 \left(\sum_{i=1}^{n_G} x_i - k\right)^2  \\
    \textbf{subject to~~} & x_i \in \{0,1\}\\
  \end{alignat*}
  The objective is maximized when $\sum_{i=1}^{n_G} x_i = k$, otherwise the objective will be negative. It's easy to see maximized value is $\binom{k}{2}$ if and only if graph $G$ has a $k$-clique. The objective can be rewritten as
\[
\sum_{i=1}^{n_G} (1+2kn_G^2) x_i - \sum_{i,j \in[n_G], i\neq j}(2kn_G^2 -1_{(i,j)\in E}) x_ix_j - n_G^2k.  
\] 

Next we are going to show that the assortment algorithm can be used to solve the following mathematical program when $\alpha_i > 0$ and $\beta_{i,j}>0$. This is more general than the above mathematical program. Therefore it will imply the assortment algorithm can be used to solve the $k$-clique problem.
 \begin{alignat*}{2}
    \textbf{maximize }   & \sum_{i=1}^{n_G} \alpha_i x_i - \sum_{i,j \in[n_G], i\neq j}\beta_{i,j} x_ix_j \\
    \textbf{subject to~~} & x_i \in \{0,1\}\\
  \end{alignat*}
  
Let $h_1< h_2 < ...< h_{n_G}$ be $n$ different positive integers such that the pairwise sums of $0,h_1,...,h_{n_G}$ are all different and $h_1 \geq 2$. We can find such sequence by using numbers that are $O(n^4)$. The existence can be proved by randomly picking integers and using probabilistic argument. Now let $a_1 < a_2 < \cdots < a_m$ be the sorted list of $\sqrt{h_1},...,\sqrt{h_{n_G}}$ and the square roots of $h_1,...,h_{n_G}$'s pairwise sums (i.e. $\sqrt{h_1+h_2}, \sqrt{h_1+h_3},...$). Here $m = n_G + \binom{n_G}{2}$. Let $b_0 = 1$ and $b_i = (a_i + a_{i+1})/2$ for $i=1,...,m-1$ and $b_m = a_m + 1$. So now we have $b_0 < a_1 < b_1 < a_2 < \cdots < a_m < b_m$. Figure \ref{fig:lb} shows what it looks with these options as lines.

\begin{figure}\begin{center}
\begin{tikzpicture}[scale=2]

\draw [->] (0,0) -- (3.5,0)  node [below] {$w$};
\draw [->] (0,-1) -- (0,1.5) node [left] {$u$};

\draw [domain=0:3.5] plot (\x, 1*\x/4-.2*5/4) node [right] {$b_{i-1}$};
\draw [domain=0:3.5,dashed] plot (\x, 1.65*\x/4-.4*5/4) node [right] {$a_i$};
\draw [domain=0:3.5] plot (\x, 2.2*\x/4-.7*5/4) node [above right] {$b_i$};

\end{tikzpicture}
\end{center}
\caption{Example of $a_i$'s and $b_i$'s viewed as lines}
\label{fig:lb}
\end{figure}
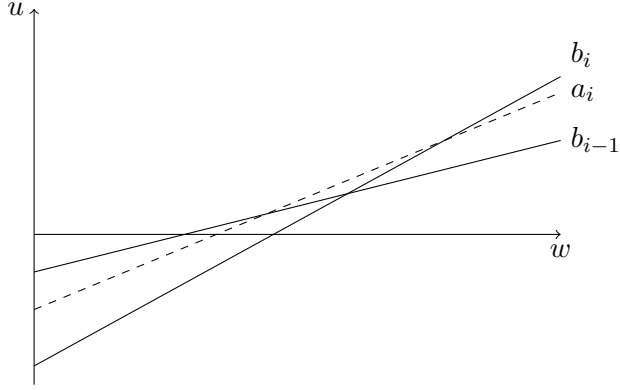

Now consider the following 2-demand valuation $v$ with $n = n_G + m + 1$ items. Each $h_i$, $1 \leq i \leq n_G$ corresponds to an item with price $h_i$ and value $\sqrt{h_i}$. Each $b_i$, $0 \leq i \leq m$ also corresponds to an item with price $b_i^2$ and value $b_i$. If two items are both corresponding to some $h_i$ and $h_j$, the value of the bundle will be $\sqrt{h_i+h_j}$. Other bundles of 2 items will have values equal to the higher value of one of its items. Since $\sqrt{h_1+h_2} < \sqrt{h_1} + \sqrt{h_2}$ and this valuation is 2-demand, its also submodular. 

It's also that no matter what the seller shows, the buyer would either buy nothing, or a single item, or 2 items that correspond to some $h_i$ and $h_j$. Therefore the sets the buyer could buy correspond to $a_i$'s and $b_i$'s. The prices of these sets are $a_i^2$ or $b_i^2$. The values of these sets are $a_i$ or $b_i$. 

Now consider a distribution $D$ on the multiplicative noise that is specified by the following CDF $F$. 
\begin{enumerate}
\item For $x = 2b_i$, $i = 0,...,m$, we have $F(x) = 1+\frac{1}{4b_m^2} - \frac{1}{x^2}$. 
\item For $ 2b_{i-1} < x < 2b_i$, $i = 1,...,m$, we have $F(x) = 1 + \frac{1}{4b_m^2} + \frac{1}{4b_{i-1}b_i}-\frac{b_{i-1}+b_i}{2xb_{i-1}b_i}$.
\item For $b_0 \leq x < 2b_0$, we have $F(x) = 1+\frac{1}{4b_m^2} - \frac{1}{x^2}$.
\item The probability of multiplicative noise smaller than $b_0$ is $\frac{1}{4b_m^2}$. It does not matter where they actually distribute at for our problem since these buyers won't buy anything anyways.
\end{enumerate}

Now define $R(x) = \left(1+ \frac{1}{b_m^2} - F(x)\right)x$ for $ b_0 \leq x \leq 2b_m$. We have
\begin{enumerate}
\item For $x =2 b_i$, $i = 0,...,m$, we have $R(x) = \frac{1}{x}$.
\item For $ 2b_{i-1} < x < 2b_i$, $i = 1,...,m$, we have $R(x) = \frac{2b_i - x}{2b_i - 2b_{i-1}} \cdot \frac{1}{2b_{i-1}} + \frac{x-2b_{i-1}}{2b_i - 2b_{i-1}}\cdot \frac{1}{2b_i}$. In other words, $R(x)$ is linear between $2b_{i-1}$ and $2b_i$. 
\item For $b_0 \leq x < 2b_0$, we have $R(x) = \frac{1}{x}$.
\end{enumerate}
Notice that $\frac{1}{x}$ is a convex function. And $R(x)$ is a piece-wise linear function based on points on curve $\frac{1}{x}$. So $R(x)$ is also convex.

Now we are going to specify the expected revenue for valuation $v$ and multiplicative noise distribution $D$. Suppose the set shown by the seller is $T$ and $C$ is the set of $a_i$'s and $b_i$'s corresponding to available set options for the buyer when $T$ is shown. Let $|C| = n_C$ and elements in $C$ are sorted as $c_1 < c_2 < \cdots < c_{n_C}$. For notation convenience, let $c_0 = 0$. It's easy to check that $c_i$ is has the highest utility for the buyer if $ c_i + c_{i-1} < x < c_i + c_{i+1}$. The revenue can be written as
\[
\left(\sum_{i=1}^{n_C-1} (F(c_i+c_{i+1}) - F(c_i + c_{i-1}))c_i^2\right) + (1- F(c_{n_G} + c_{n_G-1}))c_{n_G}^2.
\]

We are going to prove the following lemma to characterize the set $T$'s which achieve the optimal revenue.

\begin{lemma}
\label{lem:charopt}
If $T$ achieves the optimal revenue if and only if it contains all the items correspond to $b_0, ...,b_m$. 
\end{lemma}

\begin{proof}
We prove the two directions separately:
\begin{enumerate}
\item ``only if'': We prove by contradiction. Suppose $T$ does not have some item corresponds $b_i$. We will show by adding this item, the revenue will strictly increased. There are two cases:
	\begin{enumerate}
	\item When $i=m$, adding this item would increase the revenue by $(1-F(c_{n_G} + b_m))(b_m^2- c_{n_G}^2) > 0$. 
	\item When $i < m$, let $j$ be the index such that $c_j < b_i < c_{j+1}$, adding the item would increase the revenue by 
	\begin{align*}
	&& F(c_j + c_{j+1})(c_{j+1}^2 - c_j^2) - F(c_j + b_i)(b_i^2 - c_j^2) - F(c_{j+1}+b_i)(c_{j+1}^2-  b_i^2)\\
	&=& -R(c_j+c_{j+1})(c_{j+1} - c_j) + R(c_j+b_i)(b_i - c_j) + R(c_{j+1} + b_i)(c_{j+1} - b_i).
	\end{align*}
	Since $R$ is a convex function, we know 
	\[
	 R(c_j+c_{j+1}) \leq \frac{b_i-c_j}{c_{j+1}-c_j} R(c_j + b_i) + \frac{c_{j+1}-b_i}{c_{j+1}-c_j} R(c_{j+1} + b_i).
	 \]
	 Moreover, since $c_j + b_i < 2b_i $ and $c_{j+1} + b_i > 2b_i$, we know $R(c_j + b_i)$ and $R(c_{j+1} + b_i)$ locate on two different piece-wise linear functions of $R$. Therefore we know the previous inequality is strict, i.e.
	 \[
	 R(c_j+c_{j+1}) < \frac{b_i-c_j}{c_{j+1}-c_j} R(c_j + b_i) + \frac{c_{j+1}-b_i}{c_{j+1}-c_j} R(c_{j+1} + b_i).
	 \]
	 This will imply that adding item corresponds to $b_i$ would increase the revenue by some positive value. 
	\end{enumerate}
\item ``if'': Consider the option set $C$ for $T$, we will first show that if $\{b_0,...,b_m\}\subseteq C$ and $a_i \not \in C$, then $a_i$ to $C$ won't change the revenue. Similarly as the ``only if'' part, adding $a_i$ will increase the revenue by 
\[
-R(b_{i-1}+b_i)(b_{i-1} - b_i) + R(a_i+b_i)(b_i - a_i) + R(b_{i-1} + a_i)(a_i - b_{i-1}).
\]
Since $a_i +b_i \leq 2b_i$ and $a_i + b_{i-1} \geq 2b_{i-1}$, we know that $R(b_{i-1}+b_i)$, $R(a_i+b_i)$ and $R(a_i + b_{i-1})$ are on the same piece-wise linear function, therefore
\[
-R(b_{i-1}+b_i)(b_{i-1} - b_i) + R(a_i+b_i)(b_i - a_i) + R(b_{i-1} + a_i)(a_i - b_{i-1}) = 0.
\]
This means the revenue does not change.

Now we can apply this claim to $T$ several times and we know that the revenue of showing $T$ is the same the revenue of showing everything.

Thus we know that if $T$ contains all the items correspond to $b_0, ...,b_m$, no matter what other items $T$ have, the revenue is the same. This together with the ``only if'' part will imply that the ``if'' part. 
\end{enumerate}
\end{proof}

Notice that from Lemma \ref{lem:charopt}, we know that if $T$ does not contain all the items correspond to $b_0,...,b_m$, the revenue of $T$ will be strictly worse than the optimal revenue. Let's assume optimal revenue to be $OPT_D$. Let's also assume $T$'s revenue will be at least $\Delta$ worse than $OPT_D$. It's easy to see that this $\Delta$ is not too small. It should be $\Omega(1/n_G^c)$ for some constant $c$. 

Now consider another distribution $D'$ on the multiplicative noise. $D'$ will be used to embed the mathematical program. $D'$ is discrete, we will describe it by its support and density. For each $i=1,...,m$, 
\begin{enumerate}
\item If $a_i$ equals to some $\sqrt{h_j}$, $D'$ will have density $\frac{s\alpha_j}{a_i^2 - b_{i-1}^2}$ at location slightly smaller than $2a_i$.
\item IF $a_i$ equals to some $\sqrt{h_{j_1} + h_{j_2}}$, $D'$ will have density $\frac{s\beta_{j_1,j_2}}{b_i^2 - a_i^2}$ at location slightly larger than $2a_i$. 
\end{enumerate}
Here $s$ is just some scaling factor that makes sure that the densities sum up to 1. Now suppose $T$ has all the items correspond to $b_0,...,b_m$. Let $x_i$ indicate whether item corresponds to $h_i$ is included in $T$. And let $R_0$ be the revenue of only showing items correspond to $b_0,...,b_m$ on distribution $D'$, then the revenue of $T$ on $D'$ can be written as 
\[
R_0 + s(\sum_{i=1}^{n_G} \alpha_i x_i - \sum_{i,j \in[n_G], i\neq j}\beta_{i,j} x_ix_j).
\]
It is very similar to the objective of the mathematical program. 

Finally we consider a distribution $\hat{D} = (1-\rho)D + \rho D'$ where $\rho = \frac{\Delta}{2b_m^2}$. Now we are going to understand the optimal revenue for $\hat{D}$. No matter what we show, the revenue from $\rho D'$ part will be at most $\rho \cdot b_m^2 = \Delta/2$. On the other hand, for $(1-\rho)D$, if we don't show some set with all the items correspond to $b_0,...,b_m$, we will at least lose $(1-\rho)\Delta>\Delta/2$ revenue. Therefore, to achieve the optimal revenue for $\hat{D}$, we also need to include all the items correspond to $b_0,...,b_m$. Then the optimal revenue can be written as
\[
(1-\rho)OPT_D + \rho R_0 + s\rho (\sum_{i=1}^{n_G} \alpha_i x_i - \sum_{i,j \in[n_G], i\neq j}\beta_{i,j} x_ix_j).
\]
Since $OPT_D$ and $R_0$ can be computed in polynomial time, computing the optimal revenue will also solve the mathematical program. Notice that for the $k$-clique problem, the optimized value of the mathematical program is differed by some constant between the case when the instance has a $k$-clique and the case when the instance does not have a $k$-clique. It is easy to check that $\frac{\rho s}{(1-\rho)OPT_D + \rho R_0 }$ is $\Omega(1/n^c)$ for some large enough constant $c$. Therefore approximating the optimal revenue within approximation factor $1+1/n^{c'}$ for some large enough constant $c'$ will also solve the $k$-clique problem which is NP-hard.

Finally one thing worth mentioning is that in the proof we define $\hat{D}$ as a continuous distribution, it might be annoying to make it to be an input for algorithms. Actually we can discretize this distribution by picking a discrete distribution that has the same CDF as $\hat{D}$ at all locations which are pairwise sums of $0,a_1,...,a_m, b_0,...,b_m$. It's easy to check that for our valuation $v$, the preferences of options stay the same between any two adjacent locations of these locations. So discretizing $\hat{D}$ in this way won't affect the proof.
\end{proof}

\section{Missing Proofs for Well-Priced Items (Thm \ref{thm:convex}, \ref{thm:convex.constrained.gross}, Welfare)}\label{app:well-priced}
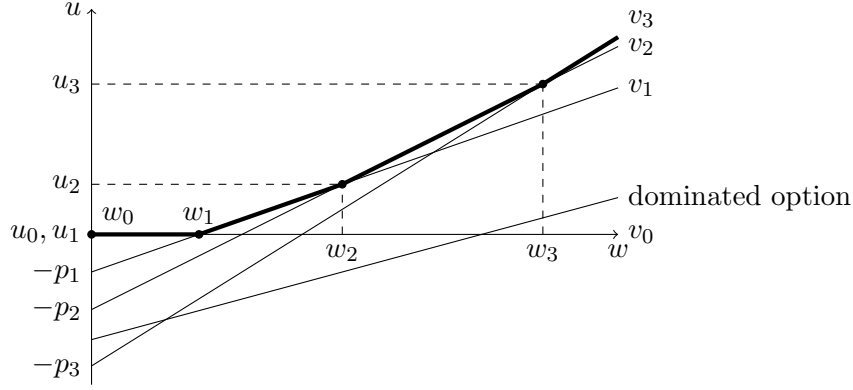
\begin{figure}\begin{center}
\begin{tikzpicture}[scale=2]

\draw [->] (0,0) -- (3.5,0)  node [below] {$w$} node [right] {$v_0$};
\draw [->] (0,-1) -- (0,1.5) node [left] {$u$};

\draw [dashed] (5/3,0) node [below] {$w_2$} -- (5/3,1/3);
\draw [dashed] (3,0) node [below] {$w_3$} -- (3,1);

\draw [dashed] (0,1/3) node [left] {$u_2$} -- (5/3,1/3);
\draw [dashed] (0,1) node [left] {$u_3$} -- (3,1);

\draw [fill] (5/3,1/3) circle [radius=0.025];
\draw [fill] (3,1) circle [radius=0.025];
\draw [fill] (1/1.4,0) circle [radius=0.025] node [above] {$w_1$};
\draw [fill] (0,0) circle [radius=0.025] node [above right] {$w_0$} node [left] {$u_0,u_1$};

\draw [domain=0:3.5] plot (\x, 1.4*\x/4-.2*5/4) node [right] {$v_1$};
\draw [domain=0:3.5] plot (\x, 2.0*\x/4-.4*5/4) node [right] {$v_2$};
\draw [domain=0:3.5] plot (\x, 2.5*\x/4-.7*5/4) node [above right] {$v_3$};

\draw (0,-.2*5/4) node [left] {$-p_1$};
\draw (0,-.4*5/4) node [left] {$-p_2$};
\draw (0,-.7*5/4) node [left] {$-p_3$};

\draw [ultra thick, domain=0:1/1.4] plot (\x, 0);
\draw [ultra thick, domain=1/1.4:5/3] plot (\x, 1.4*\x/4-.2*5/4);
\draw [ultra thick, domain=5/3:3] plot (\x, 2.0*\x/4-.4*5/4);
\draw [ultra thick, domain=3:3.5] plot (\x, 2.5*\x/4-.7*5/4);

\draw [domain=0:3.5] plot (\x, .27*\x-.7) node [right] {dominated option};

\end{tikzpicture}
\end{center}
\caption{The numbering of options in $U(T)$ for a given assortment $T$. Superscripts $T$ are dropped for illustration purposes. Notice that any dominated option is not included in the numbering.}
\label{fig:proof-notation}
\end{figure}
We start by giving an expression of the revenue obtained under a given assortment. To calculate this, we first identify all possible subsets that a buyer may purchase.

\begin{definition}
A set $S$ is \emph{dominated} for assortment $T$ if for all $w \in [0,\infty)$ there is a set $S'\subseteq T$ with $S \neq S'$, such that
$$v(S)\cdot w-\sum_{i \in S} p_i \le v(S')\cdot w-\sum_{i \in S'} p_i.$$

\noindent For an assortment $T$ we denote by $$U(T) = \{(v(S),\sum_{i\in S} p_i) :
  S \subseteq T
  \wedge
  S \text{ is not dominated for } T\}$$
the set of undominated options for $T$.
\end{definition}

The set of $U(T)$ is a totally-ordered under coordinate-wise comparisons, i.e. if $(v,p),(v',p')\in U(T)$, then either
$$(v'<v \text{ and }p'<p), \quad \text{ or } \quad
(v'=v \text{ and }p'=p), \quad \text{ or } \quad
(v'>v \text{ and }p'>p).$$
Moreover, $(0,0) \in U(T)$ and we can define the ordered sequence $(0,0)=(v^T_0,p^T_0)$,..., $(v^T_{n_T},p^T_{n_T})$ of elements of $U(T)$.
Additionally, we define as $w^T_i = \frac {p^T_i-p^T_{i-1}} {v^T_i-v^T_{i-1}}$ to be the indifference point between options $i$ and $i-1$ of $U(T)$.
The notation is illustrated in Figure~\ref{fig:proof-notation}.
Using this notation, we can write the expected revenue as follows.

\begin{lemma} \label{lem:revenue expression}
  The total revenue of assortment $T$ under distribution $F$ is equal to
  $$Rev(T) = \sum_{i=1}^{n_T} (v^T_i-v^T_{i-1}) R\left( w^T_i \right),$$
  where $R$ is the revenue function of distribution $F$, i.e. $R(w)=w(1-F(w))$.
\end{lemma}

\begin{proof}
Observe that a buyer with type $w$ prefers option $(v^T_{i},p^T_{i})$ to option $(v^T_{i-1},p^T_{i-1})$, if and only if
$w > w^T_i$.

Thus, a buyer will pay $p^T_{i}$ if his type $w \in [w^T_{i},w^T_{i+1})$, where $w^T_0 = 0$ and $w^T_{n_T+1}=+\infty$.
The total revenue is thus:
\begin{align}
Rev(T) &= \sum_{i=0}^{n_T} p^T_{i} (F(w^T_{i+1})-F(w^T_{i}))\\
&= \sum_{i=0}^{n_T}  p^T_i (1- F(w^T_{i})) - \sum_{i=1}^{n_T+1}  p^T_{i-1} (1- F(w^T_{i})) \nonumber \\
&= \sum_{i=1}^{n_T}  p^T_i (1- F(w^T_{i})) - \sum_{i=1}^{n_T}  p^T_{i-1} (1- F(w^T_{i})) \nonumber \\
&= \sum_{i=1}^{n_T}  (p^T_i-p^T_{i-1}) (1- F(w^T_{i})) = \sum_{i=1}^{n_T}  (v^T_i-v^T_{i-1}) w^T_i (1- F(w^T_{i})) \nonumber
\end{align}
which gives the required statement.
\end{proof}

\noindent An equivalent expression to Lemma~\ref{lem:revenue expression} can be written as an integral with respect to the utility
function $u_T(w)$ obtainable under an assortment $T$.

\begin{lemma} \label{lem:revenue expression integral}
  The total revenue of assortment $T$ for a revenue function $R$ supported on $[0,H]$ is
  $$Rev(T) = \int_{0}^{\infty} u_T(w) R''\left( w\right) dw - R'(H) u_T(H)$$
  where $u_T(w) = \max_{S \subseteq T} v(S) \cdot w - \sum_{i\in S} p_i$ is the utility obtained by a buyer with type $w$.
\end{lemma}
\begin{proof}
By Lemma~\ref{lem:revenue expression}, we can write the total revenue for assortment $T$ as
\begin{align*}
	Rev(T) &= \sum_{i=1}^{n_T} (v^T_i-v^T_{i-1}) R\left( w^T_i \right)
	&\text{ where } w^T_i = \frac {p^T_i-p^T_{i-1}} {v^T_i-v^T_{i-1}}\\
	&= \sum_{i=1}^{n_T} v^T_i ( R\left( w^T_i \right)  - R\left( w^T_{i+1} \right) )
	&\text{ where $w^T_{1}=0$ and $w^T_{n_T+1}=H$} \\
	&= -\sum_{i=1}^{n_T} v^T_i \int_{w^T_i}^{w^T_{i+1}} R'\left( w\right) dw \\
	&= - \int_{0}^{H} u'_T(w) R'\left( w\right) dw
	&\text{ where } u_T(w) = \max_{S \subseteq T} v(S) \cdot w - \sum_{i\in S} p_i \\
	&= \int_{0}^{H} u_T(w) R''\left( w\right) dw - R'(H) u_T(H)
	&\text{ by integration by parts since } u(0) =0.
\end{align*}
\end{proof}

\subsection{Proof of Theorem~\ref{thm:convex}}\label{sec:convex-proof}

We first show that under the conditions of Theorem~\ref{thm:convex}, the revenue function of the distribution $R(w)=w(1-F(w))$ can be well
approximated by a convex function $\hat R$. The core lemma behind the proof is thus the following:

\begin{lemma}\label{lem:convex 4 approx}
Consider a distribution $F$ and denote by $R$ its revenue function $R(w)=w(1-F(w))$. Suppose that for some $r \ge 0$:
\begin{itemize}
\item the revenue function $R$ is non-increasing for $x \ge r$, and
\item the density $f$ of $F$ is non-increasing for $x \ge r$, i.e. $F$ is concave.
\end{itemize}
Then, there exists a convex decreasing function $\hat R$ such that for all $x \ge r$, $\hat R(w) \le R(w) \le 4 \hat R(w)$.
\end{lemma}
\begin{proof}
We will show that for all $x,y\ge r$ and all $\alpha \in [0,1]$, $R(\alpha x + (1-\alpha) y) \le 4 \alpha R(x) + 4 (1-\alpha) R(y)$.
This suffices to complete the proof, as we can define $\hat R$ as the lower envelope of the function $R$, i.e.
$$\hat R(z) \triangleq \inf_{\substack{x,y\ge r, \alpha \in [0,1]\\ \alpha x + (1-\alpha) y = z}}  \alpha R(x) + (1-\alpha) R(y) \ge R(z) / 4$$
The function $\hat R$ is convex as it is the lower envelope of $R$ and satisfies $\hat R(x) \le R(x) \le 4 \hat R(x)$..

We now show that for all $ x,y \ge r$ and $\alpha \in [0,1]$ with $x \le z = \alpha x + (1-\alpha) y \le y$, it holds that $R(z) \le 4 \alpha R(x) + 4 (1-\alpha) R(y)$. The inequality holds if $\alpha \ge 1/2$ since by monotonicity of the revenue function $R(z) \le R(x) \le 4 \alpha R(x)$.
Furthermore, if $\alpha < 1/2$, we have that
$$
\begin{aligned}
4 \alpha R(x) & \ge 4 \alpha R\left( \frac {x+y} 2 \right) = 2 \alpha (x+y) \left( 1 - F\left( \frac {x+y} 2 \right) \right)& \\
&\ge 2 \alpha (x+y) \left( F(y) - F\left( \frac {x+y} 2 \right) \right) &\text{as $F(y) \le 1$}\\
&\ge 2 \alpha (x+y) \frac {F(y) - F(z)} {y-z} \frac {y-x} 2  &\text{as $F$ is concave} \\
&= (x+y) (F(y) - F(z))  &\text{as $\alpha =  \frac {y-z} {y-x}$} \\
&\ge z (F(y) - F(z)) = R(z) - z (1-F(y)) \\
&\ge  R(z) - 2 (1-\alpha) y (1-F(y)) = R(z) - 2 (1-\alpha) R(y) &\text{as $\alpha < 1/2$}
\end{aligned}
$$
Thus, for all  $\alpha \in [0,1]$, $R(z) \le 4 \alpha R(x) + 4 (1-\alpha) R(y)$.
\end{proof}

The function $\bar R$ of Lemma~\ref{lem:convex 4 approx} allows us to compute an approximation
to the optimal revenue, $\overline {Rev}(T)$. We have that,
$$\overline {Rev}(T) =
\sum_{i=1}^{n_T} (v^T_i-v^T_{i-1}) \bar R\left( w^T_i \right)
\le
\sum_{i=1}^{n_T} (v^T_i-v^T_{i-1}) R\left( w^T_i \right)
\le
4 \sum_{i=1}^{n_T} (v^T_i-v^T_{i-1}) \bar R\left( w^T_i \right)
 = 4 \overline {Rev}(T).$$

If we compute the optimal assortment with respect to function $\bar R$, i.e. $T^* = \arg\max_{T} \overline {Rev}(T)$, we will that
guarantee at least $1/4$ of the revenue of the optimal assortment for the original distribution. This is because,
$OPT = \max_{T} {Rev}(T) \le \max_{T} \overline {Rev}(T) = 4 \overline {Rev}(T^*) \le 4 {Rev}(T^*)$.

To complete the proof, we now show that the optimal assortment for function $\bar R$ is $T^* = N$.

Applying Lemma~\ref{lem:revenue expression integral} to $\bar R$ and noting that for all $T$, $u_T(w)=0$ for $w \le r$ we get that the total revenue can be written as
$$\overline{Rev}(T) =
\int_{r}^{H} u_T(w) \bar R''\left( w\right) dw - \bar R'\left( H\right) u_T(H).$$

Since $\bar R(w)$ is convex for $w \ge r$, so $\bar R''\left( w\right) \ge 0$ and since it is decreasing we have that $- \bar R'(w) \ge 0$.
Thus, in order to maximize total revenue with respect to $\bar R$ we need to maximize utility.
The assortment $T=N$ maximizes utility pointwise and guarantees the maximum total revenue with respect to $\bar R$. This implies that $T^* = N$.

\subsection{Proof of Theorem~\ref{thm:convex.constrained.gross}}\label{subsec:convex.constrained.gross}

Similar to the proof of Theorem~\ref{thm:convex}, we consider the convex approximation $\bar R$ to the function $R$ given by Lemma~\ref{lem:convex 4 approx}.
We have that its corresponding total revenue function $\overline{Rev}$ satisfies
$\overline{Rev}(T) \le {Rev}(T) \le 4\overline{Rev}(T)$ for any assortment $T$.

Thus, if we obtain an $\alpha$-approximation to the problem of maximizing $\overline{Rev}(T)$, we will obtain a $(\alpha/4)$ approximation to ${Rev}(T)$ as well.
That is, if we find a feasible assortment $T^*$, such that
$\overline{Rev}(T^*) \ge \alpha \max_{|T| \le C} \overline{Rev}(T)$, then $T^*$ satisfies:
$${Rev}(T^*)\ge\overline{Rev}(T^*) \ge \alpha \max_{|T| \le C} \overline{Rev}(T) \ge \frac{\alpha}4 \max_{T} {Rev}(T).$$
We will thus show how to efficiently approximately maximize $\overline{Rev}(T)$ when the valuation function $v(S)$ satisfies the gross-substitutes condition.

Our approach is to show that the function $\overline{Rev}(T)$ is submodular and use the fact that a simple
greedy strategy achieves a $\alpha = 1-1/e$ approximation for the problem of submodular maximization under capacity constraints.
The greedy strategy starts with the empty set and iteratively adds the item that maximizes the total revenue until the capacity
constraint is reached, i.e. if the current assortment is $T$ with $|T|<C$ the item that will be added is $\arg\max_{i \not\in T}
\overline{Rev}(T\cup\{i\})$.

To show that the function $\overline{Rev}(T)$ is submodular, we use Lemma~\ref{lem:revenue expression integral} to get that the total revenue is
$\int_{r}^{H} u_T(w) \bar R''\left( w\right) dw - \bar R'\left( H\right) u_T(H).$
As argued in Section~\ref{sec:convex-proof}, $\overline{Rev}(T)$ is a positive combination of the functions $ u_T(w) $.
In the next lemma we show that for any $w$ the function $u_T(w)$ is submodular, which implies that the function $\overline{Rev}(T)$ is submodular as well.

\begin{lemma}\label{lem:gs-submodular}
Let $v:2^N\rightarrow \mathbb{R}_{\ge0}$ be a monotone function satisfying gross-substitutes. Then, for any price vector $p$, the function
$f(T) = \max_{S\subseteq T} v(S) - \sum_{i \in S} p_i$
is submodular.
\end{lemma}
\begin{proof}
For ease of notation, we denote by the addition $S+i$ the union $S\cup\{i\}$ and $S-i$ the set difference $S\setminus\{i\}$.
To show submodularity it suffices that for all sets $T \subseteq N$ and $i,j \not\in T$ it holds that:
$f(T+i)+f(T+j) \ge f(T+i+j)+f(T)$.

To prove this, let $A \subseteq T+i+j$ and $B \subseteq T$ be the maximizers of $f(T+i+j)$ and $f(T)$ respectively.
If $j \not \in A$, then the sets $A$ and $B$ are feasible for the maximization problems $f(T+i)$ and $f(T+j)$ respectively.
This implies that $f(T+i)+f(T+j) \ge v(A) - \sum_{i \in A} p_i + v(B) - \sum_{i \in B} p_i =  f(T+i+j)+f(T)$.

Now suppose that $j \in A$. Since $v$ satisfies gross-substitutes, the M$^\sharp$-exchange property implies that either:
\begin{enumerate}
\item $v(A)+v(B)\le v(A-j) + v(B+j)$, or
\item there exists element $k \in B$ such that $v(A)+v(B)\le v(A-j+k) + v(B+j-k)$.
\end{enumerate}

In case 1, we have that $A-j$ and $B+j$ are feasible for the maximization problems $f(T+i)$ and $f(T+j)$ respectively which gives $f(T+i)+f(T+j) \ge v(A-j) - \sum_{i \in A-j} p_i + v(B+j) - \sum_{i \in B+j} p_i \le  f(T+i+j)+f(T)$.  Similarly, in case 2, we have that $A-j+k$ and $B+j-k$ are feasible for the maximization problems $f(T+i)$ and $f(T+j)$ respectively.

\end{proof}

The greedy strategy requires being able to evaluate the revenue of an assortment $S$, $\overline{Rev}(S)$, with respect to the approximate curve $\overline{R}$. This might be hard to do exactly but can be done efficiently given query access to the cumulative distribution $F$ after some preprocessing.

To derive the runtime of $\tilde O(\ell^3 n)$, we make the standard assumption that valuations and prices can be written using $B=polylog(n)$ bits so that operations can be performed efficiently. Note that the whole analysis goes through if numbers have $poly(n)$ bits resulting in polynomial runtime that depends on the exact polynomial of the bit representation.

\paragraph{Preprocessing} Let $w_{min} = \min_i \frac { p_{i} } { v(\{i\}) }$. We define the $\varepsilon$-rounded revenue curve $R^*(w)$ to be the smallest value $R(w_{min})/(1+\varepsilon)^i$ that is greater or equal to $R(w)$
for some integer $i \in [0,B/\varepsilon^2]$.

Notice that by Lemma~\ref{lem:revenue expression} we get that
$$\sum_{i=1}^{n_T} (v^T_i-v^T_{i-1}) R\left( w^T_i \right) \le Rev^*(T) \le (1+\varepsilon) \sum_{i=1}^{n_T} (v^T_i-v^T_{i-1}) R\left( w^T_i \right) + v^T_{n_T} R(w_{min})/(1+\varepsilon)^{B/\varepsilon^2}$$

Thus, $Rev(T) \le Rev^*(T) \le (1+\varepsilon) Rev(T) + \varepsilon R(w_{min})$.

Since for the optimal assortment $OPT$ we have that $OPT \ge R(w_{min})$, we get that an $\alpha$-approximate assortment for $Rev^*(\cdot)$ yields a $\alpha/(1+2\varepsilon)$ approximation to $OPT$. Therefore, from now on will assume that the underlying revenue curve is $R^*(w)$ instead of the original $R$. A query to $R^*(w)$ for a given $w$ can be computed by querying the distribution $F$ at point $w$ and the calculating $w(1-F(w))$ rounded at the appropriate multiple of $R(w_{min})/(1+\varepsilon)^i$ as described above. This preprocessing allows us to consider a revenue function $R^*$ that is piecewise constant with at most $B/\varepsilon^2$ pieces.

\paragraph{Convex Approximation} We now proceed to compute a convex curve $\bar R$ that approximates the revenue curve $R^*$. We do this by computing its convex lower envelope. To do this, we need to find the beginning and end of every segment of $R^*$. This can be done by binary search. Notice, that even though segments may start and end at arbitrary real numbers, it is easy to see from Lemma~\ref{lem:revenue expression} that only the values at points $w^T_i$ affect the revenue and those are all fractions of the form $a/b$ where both $a$ and $b$ are $B$ bit integers. There are at most $2^{2 B}$ such values so the binary search takes $O(B)$ time. 

Given the partition to $O(B/\varepsilon^2)$ intervals of the revenue curve function, we can compute the convex lower envelope in $O(B/\varepsilon^2)$ time. 
Lemma~\ref{lem:convex 4 approx} implies that at every point $w \ge w_{min}$, $\hat R(w) \le R^*(w) \le 4 \hat R(w)$.

\paragraph{Greedy Strategy} We now perform the greedy strategy which starts with the empty set $T = \emptyset$ and iteratively adds the item that maximizes the total revenue (with respect to function $\bar R$, i.e. $\arg\max_{i \not\in T}
\overline{Rev}(T\cup\{i\})$) until the capacity
constraint is reached. This achieves a 
$(1-1/e)$ approximation to $\overline{OPT} = \arg\max_{|S|\le k} \overline{Rev}(S)$, which implies a 
$(1-1/e)/4$ to ${OPT}^* = \arg\max_{|S|\le k} {Rev}^*(S)$, which in turn implies a $(1-1/e)/(4+8\varepsilon)$ approximation to $OPT$. The greedy strategy requires $O(n)$ evaluations of assortment revenue $\overline{Rev}$ to compute the best item to add to the assortment. Thus for a total of $\ell$ such steps, the total number of evaluations is $O(\ell n)$. We next show that each such revenue evaluation can be performed in time $O(\ell^2 B^2/\varepsilon^2)$. Choosing $\varepsilon = \Theta(1)$ so that $(1-1/e)/(4+8\varepsilon) = 1/6.33$ and noting that $B=polylog (n)$ we get that the total runtime is $\tilde O(\ell^3 n)$ resulting in a $1/6.33$ approximation to the optimal revenue.

\paragraph{Revenue Evaluation} We use Lemma~\ref{lem:revenue expression} to evaluate the total revenue. We have that
\begin{align*}
	&\overline{Rev}(T) = \sum_{i=1}^{n_T} (v^T_i-v^T_{i-1}) \bar R\left( w^T_i \right)
  \\
	&= \sum_{i=1}^{n_T-1} v^T_i ( \bar R\left( w^T_i \right)  - \bar R\left( w^T_{i+1} \right) ) + v^T_{n_T}  \bar R\left( w^T_{n_T} \right)
  \\
	&= \sum_{i=1}^{n_T-1} (u^T_{i+1} - u^T_{i}) \frac {\bar R\left( w^T_i \right)  - \bar R\left( w^T_{i+1} \right)}{w^T_{i+1}-w^T_i} + v^T_{n_T}  \bar R\left( w^T_{n_T} \right)
  \\
	&= \sum_{i=2}^{n_T-1} u^T_{i} \left( \frac {\bar R\left( w^T_{i-1} \right)  - \bar R\left( w^T_{i} \right)}{w^T_{i}-w^T_{i-1}} - \frac {\bar R\left( w^T_i \right)  - \bar R\left( w^T_{i+1} \right)}{w^T_{i+1}-w^T_i}\right) + 
  u^T_{n_T} \frac {\bar R\left( w^T_{n_T-1} \right)  - \bar R\left( w^T_{n_T} \right)}{w^T_{n_T}-w^T_{n_T-1}} + v^T_{n_T}  \bar R\left( w^T_{n_T} \right)
\end{align*}

Notice that the term $u^T_{i} \left( \frac {\bar R\left( w^T_{i-1} \right)  - \bar R\left( w^T_{i} \right)}{w^T_{i}-w^T_{i-1}} - \frac {\bar R\left( w^T_i \right)  - \bar R\left( w^T_{i+1} \right)}{w^T_{i+1}-w^T_i}\right)$ is 0 whenever $w^T_{i-1}$,$w^T_i$ and $w^T_{i+1}$ lie in the same segment of $R^*$, i.e. $R^*(w^T_{i-1})= R^*(w^T_i) = R^*(w^T_{i+1})$, since at these points the function $\bar R$ is linear. Therefore, even though there may be exponentially many terms in the summation, the only non-zero terms appear whenever $w^T_i$ is either the first point or the last point in the region $R^*(w) = R^*(w^T_i)$. Those points $w^T_i$ and their corresponding utilities and values can be identified by binary searching starting from the endpoint of some region to find the point where the subset that the buyer buys changes. This binary search requires $O(B)$ utility evaluations. As it is performed at most $O(B/\varepsilon^2)$ times, one for each segment of $R^*$, and utility evaluation at a given point $w$ (i.e. a demand query) for gross-substitute valuations over $\ell$ items can be performed in $O(\ell^2)$ time (see \cite{Leme} for details), the total runtime for revenue evaluation of a given assortment is $O(\ell^2 B^2/\varepsilon^2 )$.

\subsection{Welfare Maximization}
\label{app:welfare}

We will prove the following variants of our revenue results for well-priced items, applied to the objective of welfare maximization.

\begin{theorem}\label{thm:convex-welfare}
When items are well-priced, $S=N$ is a welfare-optimal unconstrained assortment.
\end{theorem}

\begin{theorem}\label{thm:convex.constrained.gross-welfare}
For gross substitutes valuations and well-priced items, a $1.6$-approximation to the welfare-optimal assortment of size at most $\ell$ can be computed in time $\tilde O(\ell^3 n)$.
\end{theorem}

Similar to Lemma~\ref{lem:revenue expression integral}, we can write the total expected welfare as an integral of utility.

We have that the expected welfare for an assortment $T$ is:

$$Wel(T) = \sum_{i=1}^{n_T} v^T_i \int_{w^T_i}^{w^T_{i+1}} f(w) dw = \int_{0}^{H} u_T'(w) f(w) dw  = f(H) u_T(H) - \int_{r}^{H} u_T(w) f'(w) dw$$

Since, the instance is well-priced, we have that the density $f$ is decreasing after $r$ and thus all coefficients of utility are positive.
This implies that optimizing utility pointwise (by showing all items) is optimal for welfare as well which completes the proof of Theorem~\ref{thm:convex-welfare}.

Moreover, if the assortment is allowed to have size at most $\ell$, we can observe that similar to the proof of Theorem~\ref{thm:convex.constrained.gross}, we can use Lemma~\ref{lem:gs-submodular} to show that the $Wel(T)$ is submodular as it is the positive combination of submodular functions ($u_T(w)$). This implies that the simple greedy algorithm that always adds the item that improves expected revenue the most yields a $(1-1/e)$-approximation. 

To calculate revenue, we first round the density function into powers of $(1+\varepsilon)$ as in the proof of Theorem~\ref{thm:convex.constrained.gross}, so that the density $f$ is piecewise constant. We then write:
\begin{align*}
&Wel(T)\\ &= \sum_{i=1}^{n_T} v^T_i (F(w^T_{i+1}) -F(w^T_{i})) \\
&= \sum_{i=1}^{n_T-1} (u^T_{i+1}-u^T_i) \frac{ F(w^T_{i+1}) -F(w^T_{i})}{w^T_{i+1}-w^T_{i}} + v^T_{n^T} (1 -F(w^T_{n^T}))\\
&= \sum_{i=1}^{n_T-1} u^T_i \left(\frac{ F(w^T_{i})-F(w^T_{i-1})}{w^T_{i}-w^T_{i-1}} - \frac{ F(w^T_{i+1}) -F(w^T_{i})}{w^T_{i+1}-w^T_{i}}\right)
+ u^T_{n^T} \frac{ F(w^T_{n^T}) -F(w^T_{n^T-1})}{w^T_{n^T}-w^T_{n^T-1}} + v^T_{n^T} (1 -F(w^T_{n^T}))
\end{align*}
As the density function is rounded into powers of $(1+\varepsilon)$, the summation contains only very few non-zero terms. These can be found and evaluated using binary search performing a utility evaluation at every step. Utility evaluations can be performed in $O(\ell^2)$ time (see \cite{Leme} for details) and thus the total runtime for computing an approximate assortment maximizing welfare is $\tilde O(\ell^3 n)$.

\section{Missing Proofs for Concave Revenue Curves (Thm~\ref{thm:concave})}
\subsection{Proof of Theorem~\ref{thm:concave}}
\label{app:concave}
  The proof of Theorem~\ref{thm:concave} uses the notation and results developed in Appendix~\ref{app:well-priced}.
  r
  Denote by $T^*$ an optimal assortment for an instance with $k$-demand buyers whose type distribution has a concave revenue curve $R$. Let $T \subseteq T^*$ be the subset that a buyer with $w=H$ purchases under assortment $T^*$, i.e. $v(T) \cdot H-\sum_{i \in T} p_i \ge \max_{S \subseteq T^*} v(S)\cdot H-\sum_{i \in S} p_i$. Since the buyer is $k$-demand, $|T|\le k$.

  Now consider the assortment $T$ instead. Under this assortment, a buyer with $w=H$ purchases the whole assortment $T$ and thus $u_T(H) = u_{T^*}(H)$.
  Moreover, since fewer options are available for purchase under assortment $T$, it holds that $u_T(w) \le u_{T^*}(w)$ for all $w \in [0,H]$.

  By Lemma~\ref{lem:revenue expression integral} and by noting that $R''(w) \le 0$ for $w \in [0,H]$ since $R$ is concave, we have that
  $$Rev(T) = \int_{0}^{\infty} u_T(w) R''\left( w\right) dw - R'(H) u_T(H) \ge \int_{0}^{\infty} u_{T^*}(w) R''\left( w\right) dw - R'(H) u_{T^*}(H) = Rev(T^*)$$

  This implies that $T$ with size at most $k$ is an optimal assortment.
 
\section{Missing Proofs for Additive $k$-Demand Buyers (Thm~\ref{thm:dpmain}, \ref{thm:dpmainconstrained})}
\subsection{The DP for additive $k$-demand valuations}
\label{sec:dp}
In this section we will write $v(\{i\})$ as $v_i$. And since the valuation is additive $k$-demand, for each subset $S$ of size at most $k$, $v(S) = \sum_{i\in S} v_i$. 

We are going to characterize each item as a line in two dimensional space. The $x$-axis is $w$ and the $y$-axis is the utility of the item for buyer with some $w$. So for item $i$, the slope of the line will be $v_i$ and it also passes point $(0,-p_i)$. See Figure \ref{fig:uv} as an example.

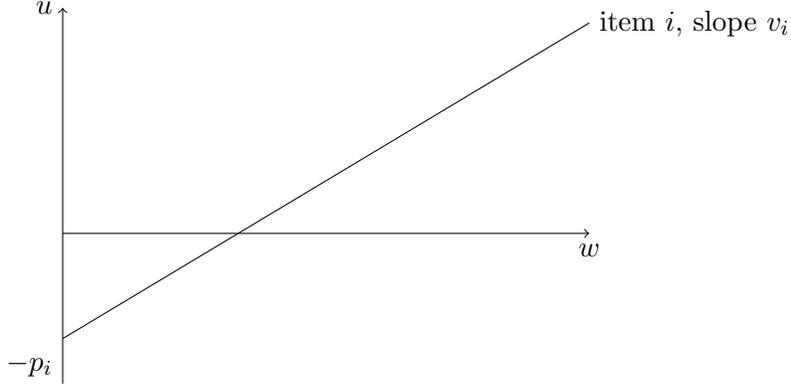
\begin{figure}\begin{center}
\begin{tikzpicture}[scale=2]

\draw [->] (0,0) -- (3.5,0)  node [below] {$w$};
\draw [->] (0,-1) -- (0,1.5) node [left] {$u$};

\draw (0,-.7*5/4) node [left] {$-p_i$};

\draw [domain=0:3.5] plot (\x, .6*\x-.7) node [right] {item $i$, slope $v_i$};

\end{tikzpicture}
\end{center}
\caption{Example of an item option viewed as a line}
\label{fig:uv}
\end{figure}
We are going to assume all the lines are in the general position (i.e., no two lines are the same, no three lines intersect at the same point, no two intersection points have the same $w$ value, and no two items have exactly the same price). We will also assume the density of $w$ at intersection points are negligible.

Now back to the original problem, the seller just need to choose $l$ lines in the plane to show to the buyer. For notation convenience, the seller will always include $2k-1$ items with prices $0$ and values $0$ (to mean not buying). And the buyer with some $w$ will buy the lines that have top-$k$ utilities at $w$. 

When $k =1$ (i.e. unit-demand case), the problem can be solved by a simple DP: scan from small $w$ to large $w$ and remember which item with the largest utility. When $k>1$, only remembering which items are in top-$k$ might not be enough. In the following figure, item 1 and 2 are top-2 at $w_1$, item 1 and 3 are top-2 at $w_2$ and item 2 and 3 are top-2 at $w_3$. Assume the seller already include item 1 and 3 and the buyer is 2-demand. Then whether item 2 shows up in top-2 at $w_1$ is correlated with item 2 shows up in top-2 at $w_3$. However, if we just scan from $w_1$ to $w_3$ and remember only the top-2, then at $w_2$ we will forget whether the seller includes item 2 or not. See Figure \ref{fig:k=2} as an example.

\begin{figure}\begin{center}
\begin{tikzpicture}[scale=2]

\draw [->] (0,0) -- (3.5,0)  node [below] {$w$};
\draw [->] (0,-1) -- (0,1.5) node [left] {$u$};

\draw [dashed] (1.3,0) node [below] {$w_1$} -- (1.3,0.31);
\draw [dashed] (1.93,0) node [below] {$w_2$} -- (1.93,0.57);
\draw [dashed] (3,0) node [below] {$w_3$} -- (3,1.36);

\draw [fill] (1.93,0.57) circle [radius=0.025];
\draw [fill] (3,1.36) circle [radius=0.025];
\draw [fill] (1.3,0.31) circle [radius=0.025];

\draw [domain=0:3.5] plot (\x, 1.7*\x/4-.2*5/4) node [right] {$v_1$};
\draw [domain=0:3.5] plot (\x, 2.3*\x/4-.6) node [right] {$v_2$};
\draw [domain=0:3.5] plot (\x, 3*\x/4-.7*5/4) node [above right] {$v_3$};

\draw (0,-.2*5/4) node [left] {$-p_1$};
\draw (0,-.6) node [left] {$-p_2$};
\draw (0,-.7*5/4) node [left] {$-p_3$};

\end{tikzpicture}
\end{center}
\caption{Example of problems of only remembering top-$k$ when $k=2$ }
\label{fig:k=2}
\end{figure}
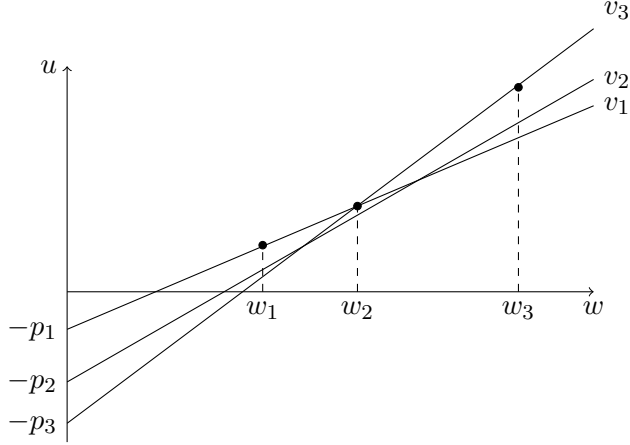

We will first prove the following lemma to give some characterization of the top-$k$ in this two dimensional space. 
\begin{lemma}
\label{lem:2k-1}
Let $w_1 < w_2 < w_3$, suppose item $i$ is not in the top-$(2k-1)$ at $w_2$, then one of the followings is true:
\begin{enumerate}
\item Item $i$ is not in the top-$k$ at $w_1$.
\item Item $i$ is not in the top-$k$ at $w_3$.
\end{enumerate}
\end{lemma}

\begin{proof}
By symmetry we will just show if item $i$ is in top-$k$ at $w_1$ then item $i$ is not in the top-$k$ at $w_3$. Since item $i$ is in the top-$k$ at $w_1$ and item $i$ is not in the top-$(2k-1)$ at $w_2$, there are at least $k$ items such that they have lower utility than item $i$ at $w_1$ and higher utility than item $i$ at $w_2$. It means these items have larger slope than item $i$ and therefore they will also have higher utility than item $i$ at $w_3$. Thus item $i$ is not in the top-$k$ at $w_3$.
\end{proof}

Intuitively, Lemma \ref{lem:2k-1} says that although top-$k$ items might not be enough information for the DP, top-$k$ items together with $k-1$ extra items might be enough information for the DP.

Now we are going to specify the DP procedure (for $l=n$):

\begin{enumerate}
\item Consider all the intersection points between any two lines, sort them by their $w$'s, let the positive such $w$'s be $w_1< w_2< \cdots<w_N$. $N = O(n^2)$. Let $w_{N+1} = +\infty$. For notation convenience, define the rightmost intersection point to be $w_{N+1}$. 
\item The state of the DP will have three components $(w,S,T)$:
	\begin{enumerate}
 	\item $w$ means the current intersection point's $w$ value. (The DP is scanning from small $w$ to large $w$)
 	\item $S$ is a set of $k$ items meaning the top-$k$ items at $w$.
 	\item $T$ is a set of $k-1$ items.
 	\item The DP value $DP(w,S,T)$ means the maximum revenue of buyers with multiplicative noise at most $w$ achieved  by state $(w,S,T)$.
 	\end{enumerate}
\item The initial state has $w = 0$, $S$ to be the set of $k$ items with prices 0 and values 0, $T$ to be the set of $k-1$ items with prices 0 and values 0. The DP value of this initial state is set to be 0. 
\item Now we specify the transition. For $w = 0,w_1,w_2,...,w_N$:
	\begin{enumerate}
	\item Enumerate all achievable state $(w,S,T)$. 
	\item For each such state let item $i_1$ be the $k$-th item at $w$ inside set $S$. Let item $i_2$ be the $(k-1)$-th item at $w$ inside $T$. Let $w_1$ be item $i_1$'s next intersection point's $w$ value. Let $w_2$ be item $i_2$'s next intersection point's $w$ value. Assume the intersection at $w_1$ is between item $i_1$ and item $j_1$ and the intersection at $w_2$ is between item $i_2$ and $j_2$. 
		\begin{enumerate}
		\item If $w_1 = w_2 = w_{N+1}$, update($w_1,S,T, DP(w,S,T)+(1-F(w))\cdot \sum_{i\in S} p_i$).
		\item If $w_1 < w_2$ and $j_1 \in T$, let $S' = S \cup \{j_1\} \backslash \{i_1\}$, and $T' = T \cup \{i_1\} \backslash \{j_1\}$, \\update($w_1,S',T', DP(w,S,T)+(F(w_1)-F(w))\cdot \sum_{i\in S} p_i$).
		\item If $w_1 < w_2$ and $j_1 \not\in T$, update($w_1,S,T, DP(w,S,T)+(F(w_1)-F(w))\cdot \sum_{i\in S} p_i$).
		\item If $w_1 > w_2$, update($w_2,S,T, DP(w,S,T)+(F(w_2)-F(w))\cdot \sum_{i\in S} p_i$).
		\item If $w_1 > w_2$ and $j_2 \not \in S \cup T$ and $p_{j_2} > p_{i_2}$, let $T' = T \cup \{j_2\} \backslash \{i_2\}$,\\ update($w_2,S,T', DP(w,S,T)+(F(w_2)-F(w))\cdot \sum_{i\in S} p_i$). 
		\end{enumerate}
	\item The update($w,S,T,value$) is just to do the following:
	\begin{enumerate}
		\item Mark state $(w,S,T)$ as achievable.
		\item If state $(w,S,T)$ already has some value $DP(w,S,T)$ and $value > DP(w,S,T)$, set $DP(w,S,T) \leftarrow value$. 
	\end{enumerate}
	\end{enumerate}
\item The DP will output the optimal revenue as the maximum DP value of all states with $w = w_{N+1}$. Let $P$ be the DP-path from the initial state to the state achieves this DP value. Here DP-path is defined as a directed path on the states and there's an directed edge between two states if one state can transit to another state. The optimal assortment will be all the items inside some $S$ on DP-path $P$.
\end{enumerate}

The following lemma discusses the running time of the DP algorithm. Notice that, as the proof suggests, although there are $\Omega(n^{2k+1})$ possible states, the number of achievable states is $O(n^{2k})$.
\begin{lemma}
\label{lem:dptime}
The above DP procedure (for $l=n$) runs in time $O(\max\{ n^{2k},n^2\log(n)\})$.
\end{lemma}

\begin{proof}
For each achievable state $(w,S,T)$, we know that $|S| \leq k$, $|T| \leq k-1$ and $w$ has $O(n^2)$ different values. Also if $w \neq w_{N+1}, 0$ then one of the two lines intersected at $w$ will be inside set $S$ or $T$. therefore the number of achievable states with $w \neq w_{N+1}$ is at most $O(n^{2k})$. On the other hand, the number of states with $w = w_{N+1}$ is at most $O(n^{2k-1})$. Thus in total there are at most $O(n^{2k})$ achievable states. If we preprocess the sorted list of intersections of each item in $O(n^2\log(n))$ time, we can do the state transition in $O(1)$ time. Therefore the total running time is $O(\max\{ n^{2k},n^2\log(n)\})$.

\end{proof}

\begin{lemma}
\label{lem:dppath}
For any DP-path $(0,S_0,T_0)\rightarrow (w_{i_1},S_{i_1},T_{i_1}) \rightarrow (w_{i_2},S_{i_2},T_{i_2}) \rightarrow \cdots \rightarrow (w_{i_t},S_{i_t}, T_{i_t})$ with $w_{i,t} = w_{N+1}$. Let $S = \cup_{j=1}^t S_{i_j}$. For any $w_{i_j} < w < w_{i_{j+1}}$, the top-$k$ utility items in $S$ at $w$ is $S_{i_j}$. 
\end{lemma}

\begin{proof}
Let $S'$ be the top-$k$ utility set at $w$ for $w_{i_j} < w < w_{i_{j+1}}$. We will show for each item $i \in S'$, $i \in S_{i_j}$. Since $|S'| = |S_{i_j}| = k$, this implies $S' = S_{i_j}$.
Let's prove by contradiction, suppose there exists some $i \in S'$ and $i \not\in S_{i_j}$. There are two cases:
\begin{enumerate}
\item $i \in T_{i_j}$: First of all, by the DP transition procedure, it's easy to check that every item in $S_{i_j}$ has higher utility than every item in $T_{i_j}$ for buyers at $w_{i_j}$. Again by the DP transition procedure, this even holds for $w_{i_j} < w < w_{i_{j+1}}$. Therefore, there are at least $k$ items in $S$ that has utility higher than $i$ at $w$. This contradicts with the fact that $i \in S'$.
\item $i \not\in T_{i,j} \cup S_{i_j}$: Since $i \in S' \subseteq S$, we know there exists $x$ such that $i \in S_{i_x}$. The following argument is very similar to Lemma \ref{lem:2k-1} but slightly different:
\begin{enumerate}
\item If $x < j$, there exists $x < y < j$ when item $i$ leaves the sets of DP states ($S\cup T$). As item $i$ has higher utility than every item in $T_{i_x}$ for buyers at $w_{i_x}$, item $i$ is surpassed by at least $k$ items when we scan $w$ from $x$ to $y$. These items will have higher values and slopes than item $i$ at $w_{i_j}$. Some of them might not stay in $S\cup T$ at $w_{i_j}$, but it means they are surpassed by some other items with even higher values and slopes. Therefore all items in $S_{i_j}$ will have higher values than item $i$ at $w_{i_j}$. This contradicts with the fact that $i \in S'$. 
\item If $x > j$, there exists $j < y < x$ when item $i$ joins the sets of DP states ($S\cup T$). As item $i$ has higher utility than every item in $T_{i_x}$ for buyers at $w_{i_x}$, item $i$ surpasses at least $k$ items when we scan $w$ from $y$ to $x$. These items will have higher values and lower slopes than item $i$ at $w_{i_j}$. Some of them might not stay in $S\cup T$ at $w_{i_j}$, but it means they are surpass by some other items with lower slopes before item $i$ surpasses them. Therefore all items in $S_{i_j}$ will have higher values than item $i$ at $w_{i_j}$. This contradicts with the fact that $i \in S'$. 
\end{enumerate}
\end{enumerate} 
\end{proof}

\begin{lemma}
\label{lem:dpvalue}
There exists a DP-path that achieves the optimal revenue. 
\end{lemma}

\begin{proof}
Let $S^*$ be the smallest set the seller shows to the buyer to achieve the optimal revenue. Since $S^*$ is the smallest such set, every item $i$ in $S^*$ must be in the top-$k$ utility set for some buyer. 

Notice that in the DP, the DP-path starts at the initial state and only diverges at step 4(b)iv and 4(b)v. Basically for $w_1 > w_2$ and $j_2 \not \in S \cup T$ and $p_{j_2} > p_{i_2}$, the DP-path can choose to go to either 4(b)iv or 4(b)v. Now consider the DP-path that chooses 4(b)v only when $j_2 \in S^*$. Let this DP-path be $(0,S_0,T_0)\rightarrow (w_{i_1},S_{i_1},T_{i_1}) \rightarrow (w_{i_2},S_{i_2},T_{i_2}) \rightarrow \cdots \rightarrow (w_{i_t},S_{i_t}, T_{i_t})$ with $w_{i,t} = w_{N+1}$. Define $S' = \cup_{j=1}^t S_{i_j}$. 

It's clear that only items in $S^*$ (except those 0 value items in the initial state) can join the sets in the DP states. So we have $S' \subseteq S^*$. 

For each item $i$ in $S^*$, we know $i$ is in the top-$k$ utility set of $S^*$ for some buyer. Let's assume it is in the top-$k$ utility set of $S^*$ for some buyer at $w$ such that $w_{i_j} < w < w_{i_{j+1}}$. Then by similar argument as in Lemma \ref{lem:dppath}, we know item $i$ is also in the top-$k$ utility set of $S^*$ at $w_{i_j}$. We want to show $i \in S_{i_j}$ by contradiction:
\begin{enumerate}
\item $i \in T_{i_j}$: Similarly as the argument in Lemma \ref{lem:dppath}, we know every item in $S_{i_j}$ has higher utility than every item in $T_{i_j}$ for buyers at $w_{i_j}$. Since $S_{i_j} \subseteq S^*$, this contradicts with the fact $i$ is also in the top-$k$ utility set of $S^*$ at $w_{i_j}$.
\item $i \not\in T_{i,j} \cup S_{i_j}$: Define $r_p$ as item $i$'s rank together with items in $T_{i,p} \cup S_{i,p}$. We know $r_0 > 2k-1$ and $r_j \leq k$. Let $p'$ be the largest $p' < j$ such that $r_{p'} = 2k-1$. Such $p'$ exists and at $w_{i_{p'}}$, item $i$ should be in $S_{i_{p'}} \cup T_{i_{p'}}$ because of the way we choose the DP-path. And for $ p'<p\leq j$, we know $r_p \leq 2k-1$, and therefore we know item $i$ stays in $S_{i_p} \cup T_{i_p}$ for $ p'<p\leq j$. Now we get a contradiction.
\end{enumerate}

Since each item $i$ in $S^*$ is also in $S'$ and $S' \subseteq S^*$, we know $S' = S^*$. By Lemma \ref{lem:dppath} together with how the DP updates values, we know the above DP-path achieves the optimal revenue. 
\end{proof}

Combining Lemma \ref{lem:dptime}, Lemma \ref{lem:dppath} and Lemma \ref{lem:dpvalue}, we get the following result:
\begin{theorem}[Restatement of Theorem \ref{thm:dpmain}]
\label{thm:dp}
For additive $k$-demand valuations, there exists an algorithm with $O(\max\{ n^{2k},n^2\log(n)\})$ running time (for $l=n$) which finds the optimal revenue and the corresponding assortment in the offline setting and the noisy case.
\end{theorem}

\begin{proof}
We will just use the DP procedure discussed above. Lemma \ref{lem:dptime} guarantees the running time. 

By Lemma \ref{lem:dpvalue}, we know our DP will output some value at least the optimal revenue. On the other hand, by Lemma \ref{lem:dppath} together with how the DP updates values, we know any DP-path will achieve some value equal to the revenue when the seller showing some set. Therefore our DP will output some value at most the optimal revenue. Therefore, out DP will output the optimal revenue. 

Let this optimal revenue be achieved by some DP-path $P$. The assortment outputted by the DP would be all the items in some $S$ of the DP state on $P$. By Lemma \ref{lem:dppath}, we know this assortment achieves the DP value of the path which is just the optimal revenue. 
\end{proof}

\begin{corollary}[Restatement of Theorem \ref{thm:dpmainconstrained}]
\label{cor:dp}
The above DP procedure can be extended to the case for arbitrary $l$ and the extended version runs in time $O(n^{2k} l)$.
\end{corollary}

\begin{proof}
To extend the DP procedure, we just add to the DP state a counter which counts the number of transition 4(b)v on the DP-path. This initial state has this counter 0 and we only care about DP state with this counter at most $l$. This will result in a factor $l$ blow up in the running time, i.e. $O(n^{2k})$ to $O(n^{2k} l)$. (It does not affect the preprocessing time $O(n^2 \log(n))$.)
 
For correctness, we know that each item can enter $S$ of the DP state only after it enters $T$ of the DP state by transition 4(b)v. By Lemma \ref{lem:dppath} and the DP procedure, it is easy to see that such counter restriction makes sure that each DP-path's corresponding assortment has size at most $l$. 

Now let $S^*$ be the optimal assortment of size at most $l$, we need to show there's a DP-path that achieves the optimal revenue, and transition 4(b)v is used $|S^*|$. Recall that the DP-path is decided by whether we take each transition 4(b)v. Consider the following DP-path. For each item $i \in S^*$, let $w$ be the small multiplicative noise such that $i$ appears in the top-$k$ utility set. We let item $i$ enter the set $S\cup T$ before that. By similarly argument as Lemma \ref{lem:dppath}, we know that this will give the optimal revenue. On the other hand, this DP-path uses transition 4(b)v $|S^*|$ times. So it will be found by the extended version of the DP. Therefore the extended version of the DP will correctly output the optimal revenue.
\end{proof}
\begin{remark}
\label{rm:dp}
The DP works not only for maximizing revenue but also maximizing the expectation of any function $g(S)$ of the items bought. In the case of revenue, $g(S) = \sum_{i\in S} p_i$ while in the case of welfare $g(S)=v(S)$. This is because the DP tracks optimizes over all possible subsets of items bought in different regions of the buyer's typespace. We just need to switch the target objective from $\sum_{i\in S} p_i$ (which is the revenue gained when the buyer buys set $S$) in the DP to other objectives. 
\end{remark}

\section{Missing Proofs for Learning from Demand Samples} 
\label{app:learn}
\subsection{The DP in the learning setting for additive $k$-demand valuations}

In this section, we are going to show that the DP of Theorem \ref{thm:dpmain} and Theorem \ref{thm:dpmainconstrained} (for $l=n$ or arbitrary $l$) can be extended to the learning setting. For the learning setting, we are going to assume that for any two intersection points with $w$ value $w_1$ and $w_2$, then $|F(w_1) - F(w_2)| > \varepsilon_0$.

\begin{theorem}
\label{thm:dplearn}
Let $\varepsilon \leq \varepsilon_0$. The DP of Theorem \ref{thm:dpmain} and Theorem \ref{thm:dpmainconstrained} can be extended to the learning setting. The DP's output is within $O(\varepsilon \cdot \max_{|S|=k} \sum_{i \in S} p_i)$ additive factor to the optimal revenue and the DP has sample complexity $\Theta(n^{k+1}\log(n)/\varepsilon^2)$.
\end{theorem}

\begin{proof}

The first thing to notice is that, in the DP, the only information we used about  $v$ and $\mathcal{F}$ is about the CDF values of $w$'s of intersection points. Furthermore, we are only looking at intersection points that have positive utility and rank at least $k$ and intersection points between $x$-axis and other items. Consider the following sampling procedures:

\begin{enumerate}
\item For each item $i$, show it (just a single item) to $\Theta(\log(n)/\varepsilon^2)$ buyers. Let $\hat{X}_i$ be the empirical probability of the buyer not buying the item. Set the estimated CDF value of the $w$ of the intersection point between $x$-axis and item $i$ to be $\hat{X}_i$.
\item For each subset $S$ of size $k+1$, show it to $\Theta(\log(n)/\varepsilon^2)$ buyers. Let $\hat{X}_i$ be the empirical probability of buyers buying set $S \backslash \{ i \}$.  Let $\hat{X}$ be the probability of buyers buying fewer than $k$ items. For each $\hat{X}_i > 0$, let $j_i$ be the largest $j$ such that $\hat{X}_j > 0$ and $ j < i$. If such $j_i$ exists, then set the estimated CDF value of the $w$ of the intersection point between item $i$ and item $i_j$ to be $\hat{X}+ \sum_{j<i} \hat{X}_j$. 
\end{enumerate}

\begin{lemma}
For any constant $c$, with probability at most $1- 1/n^c$, we can learn the CDF values of $w$'s of the intersection points used in the DP within additive error $\varepsilon$ by using $O(n^{k+1}\log(n)/\varepsilon^2)$ samples. 
\end{lemma} 

\begin{proof}
We will just use the procedure described above. It's clear that this procedure uses \\$O(n^{k+1}\log(n)/\varepsilon^2)$ samples. 

Now we are going to show the approximate guarantee. We are going to only show this for the second part of the sampling procedure since similar proof can be applied to the first part. First of all, define $X$ and $X_i$'s to be actual probabilities. If we set $\hat{X} \leftarrow X$ and $\hat{X}_i = X_i$ for $i \in [n]$, it's easy to check that we will find all the intersection points which have positive utility and rank at least $k$ and we will also get the CDF values correct. By Chernoff bound and union bound, we can show that with probability $1 - 1/n^c$, for each set $S$ of $k+1$, we have $|\hat{X}+ \sum_{j<i} \hat{X}_j - (X+ \sum_{j<i} X_j)| < \varepsilon$ and $|\hat{X}_i - X_i| < \varepsilon$ for all $i \in [n]$ if we show each $S$ to $c' \log(n)/\varepsilon^2$ buyers for some large enough constant $c'$. By the assumption in the beginning of this section we know $X_i$ is either at least $\varepsilon_0$ or equal to 0. This together with $|\hat{X}_i - X_i| < \varepsilon \leq \varepsilon_0$ implies that $\hat{X}_i > 0$ if and only if $X_i > 0$. Then we know our sampling procedure will correctly find all the intersection points we want.  Furthermorer, since we have $|\hat{X}+ \sum_{j<i} \hat{X}_j - (X+ \sum_{j<i} X_j)| < \varepsilon$, we know the CDF values are approximated within additive error $\varepsilon$. 
\end{proof}

Now we have all the intersection points we need for the DP and they are also sorted correctly with respect to their $w$ values. Therefore if we run the DP on the sampling result, we will get same set of DP paths as in the offline setting. Now we will show each DP path's DP value is approximated within some additive error by the following lemma. And this will conclude the proof of the theorem.

\begin{lemma}
For any DP path, the DP value using sampled data is within additive error $\varepsilon \cdot \max_{|S|=k} \sum_{i \in S} p_i$ compared to the DP value using the accurate data.  
\end{lemma}

\begin{proof}
Assume the DP path has $m$ states. Let the $i$-th DP state's $w$ has CDF value $F_i$. Let $s_i$ be the total price of items bought by buyers whose multiplicative noise is between the $i$-th DP state and the $i+1$-th DP state. It's to see that $s_1 \leq s_2 \leq \cdots s_m$ as a buyer with smaller multiplicative noise would not prefer to buy more expensive items. 

We can write the DP value using the accurate data as
\begin{eqnarray*}
&& F_1 \cdot 0 + (F_2 - F_1) \cdot s_1  + \cdots  + (F_m - F_{m-1}) \cdot s_{m-1} + (1-F_m) \cdot s_m  \\ 
&=& F_1 \cdot (0- s_1) +  F_2 \cdot (s_1 - s_2) + \cdots + F_m\cdot (s_{m-1}- s_m) + 1 \cdot s_m \\
\end{eqnarray*}
Now let the estimated CDF values to be $\hat{F}_i$. The DP value using sampled data is 
\[
\hat{F}_1 \cdot (0- s_1) +  \hat{F}_2 \cdot (s_1 - s_2) + \cdots + \hat{F}_n\cdot (s_{m-1}- s_m) + 1 \cdot S_m.
\]

We know that $|F_i - \hat{F}_i| < \varepsilon$.  So the difference between the DP value using sampled data and the DP value using accurate data is
\begin{eqnarray*}
&&(F_1-\hat{F}_1) \cdot (0- s_1) +  (F_2-\hat{F}_2) \cdot (s_1 - s_2) + \cdots + (F_m-\hat{F}_m)  \cdot (S_{m-1}- S_m) \\
&\leq& |(F_1-\hat{F}_1) \cdot (0- s_1)| +  |(F_2-\hat{F}_2) \cdot (s_1 - s_2)| + \cdots + |(F_m-\hat{F}_m)  \cdot (S_{m-1}- S_m)| \\
&\leq& \varepsilon \cdot (s_1 + (s_2 - s_1) + \cdots + (s_m-s_{m-1})) \\
&\leq& \varepsilon \cdot s_m \leq \varepsilon \cdot \max_{|S|=k} \sum_{i \in S} p_i.
\end{eqnarray*}
\end{proof}
\end{proof}

\end{document}